\documentclass{article}
\usepackage{amsfonts}
\usepackage{amsmath}
\usepackage{amssymb}
\usepackage{float}
\usepackage{geometry}
\usepackage{graphicx}
\usepackage{epstopdf}
\usepackage{bm}
\usepackage{stackrel}
\usepackage[T1]{fontenc}
\usepackage[utf8]{inputenc}
\usepackage{authblk}
\usepackage{color}
\usepackage{hyperref}

\newtheorem{theorem}{Theorem}[section]

\newtheorem{proposition}[theorem]{Proposition}
\newtheorem{remark}[theorem]{Remark}

\newenvironment{proof}[1][Proof]{\noindent\textbf{#1.} }{\ \rule{0.5em}{0.5em}}
\newcommand{\comment}[1]{}
\newcommand{\EEA}{\end{eqnarray}}
\newcommand{\BEA}{\begin{eqnarray}}

\providecommand{\keywords}[1]{\textbf{{Keywords.}} #1}

\begin{document}
\title{Parameter Estimation with Data-Driven Nonparametric Likelihood Functions}
\author[a]{Shixiao W. Jiang}
\author[a,b,c]{John Harlim \thanks{Corresponding author: jharlim@psu.edu}}
\affil[a]{Department of Mathematics, the Pennsylvania State University, 109 McAllister Building, University Park, PA 16802-6400, USA}
\affil[b]{Department of Meteorology and Atmospheric Science, the Pennsylvania State University, 503 Walker
Building, University Park, PA 16802-5013, USA}
\affil[c]{Institute for CyberScience, the Pennsylvania State University, 224B Computer Building, University Park, PA 16802, USA}
\date{\today}
\maketitle

\begin{abstract}
In this paper, we consider a surrogate modeling approach using a data-driven nonparametric likelihood function constructed on a manifold on which the data lie (or to which they are close). The proposed method represents the likelihood function using a spectral expansion formulation known as the kernel embedding of the conditional distribution. To respect the geometry of the data, we employ this spectral expansion using a set of data-driven basis functions obtained from the diffusion maps algorithm. The theoretical error estimate suggests that the error bound of the approximate data-driven likelihood function is independent of the variance of the basis functions, which allows us to determine the amount of training data for accurate likelihood function estimations. Supporting numerical results to demonstrate the robustness of the data-driven likelihood functions for parameter estimation are given on instructive examples involving stochastic and deterministic differential equations. When the dimension of the data manifold is strictly less than the dimension of the ambient space, we found that the proposed approach (which does not require the knowledge of the data manifold) is superior compared to likelihood functions constructed using standard parametric basis functions defined on the ambient coordinates. In an example where the data manifold is not smooth and unknown, the proposed method is more robust compared to an existing polynomial chaos surrogate model which assumes a parametric likelihood, the non-intrusive spectral projection. In fact, the estimation accuracy is comparable to direct MCMC estimates with only eight likelihood function evaluations that can be done offline as opposed to 4000 sequential function evaluations, whenever direct MCMC can be performed. A robust accurate estimation is also found using a likelihood function trained on statistical averages of the chaotic 40-dimensional Lorenz-96 model on a wide parameter domain.
\end{abstract}

\

\keywords{Bayesian inference, MCMC, diffusion maps, nonparametric likelihood function, surrogate modeling, reproducing kernel Hilbert space, kernel embedding of the conditional distribution}




\section{\label{sec:level1}Introduction}

Bayesian inference is a popular approach for solving inverse problems with
far-reaching applications, such as parameter estimation and uncertainty
quantification (see for example \cite{ks:05,sullivan2015,DS15}). In this
article, we will focus on a classical Bayesian inference problem of
estimating the conditional distribution of hidden parameters of dynamical
systems from a given set of noisy observations. In particular, let $\mathbf{x%
}(t;\bm{\theta })$ be a time-dependent state variable, which implicitly
depends on the parameter $\bm{\theta }$ through the following initial value
problem,
\begin{equation}
\mathbf{\dot{x}}=f(\mathbf{x},\bm{\theta }),\text{ \ \ }\mathbf{x}(0)=%
\mathbf{x}_{0}. \label{Eqn:ode}
\end{equation}%
Here, for any fixed $\bm{\theta }$, $f$ can be either deterministic or
stochastic. Our goal is to estimate the conditional distribution of $%
\bm{\theta }$, given discrete-time noisy observations $\mathbf{y}^{\dag }=\{%
\mathbf{y}_{1}^\dag,\ldots ,\mathbf{y}_{T}^\dag\}$, where:%
\begin{equation}
\mathbf{y}_{i}^\dag=g(\mathbf{x}_{i}^{\dag },\bm{\xi }_{i}),\text{ \ \ }%
i=1,\ldots ,T. \label{Eqn:yi_xi}
\end{equation}%
Here, $\mathbf{x}_{i}^{\dag }\equiv \mathbf{x}(t_{i};\bm{\theta }^{\dag })$
are the solutions of Equation (\ref{Eqn:ode}) for a specific hidden parameter $%
\bm{\theta }^{\dag }$, $g$ is the observation function, and $\bm{\xi }_{i}$
are unbiased noises representing the measurement or model error.
Although the proposed approach can also estimate the conditional density of the initial
condition $\mathbf{x}_{0}$, we will not explore this inference problem in
this article.

Given a prior density, $p_{0}(\bm{\theta })$, Bayes' theorem states that the
conditional distribution of the parameter $\bm{\theta}$ can be estimated as,%
\begin{equation}
p(\bm{\theta }|\mathbf{y}^{\dag })\propto p(\mathbf{y}^{\dag }|\bm{\theta }%
)p_{0}(\bm{\theta }), \label{Eqn:bayes}
\end{equation}%
where $p(\mathbf{y}^{\dag }|\bm{\theta })$ denotes the likelihood function
of $\bm{\theta }$ given the measurements $\mathbf{y}^\dag$ that {depend} on a hidden
parameter value $\bm{\theta }^{\dag }$ through (\ref{Eqn:yi_xi}). In most applications, the
statistics of the conditional distribution $p(\bm{\theta }|\mathbf{y}^{\dag
})$ are the quantity of interest. For example, one can use the mean
statistic as a point estimator of $\bm{\theta }^{\dag }$ and the higher order moments
for uncertainty quantification. To realize this goal, one draws samples of $%
p(\bm{\theta }|\mathbf{y}^{\dag })$ and estimates these statistics via
Monte Carlo averages over these samples. In this application, Markov Chain
Monte Carlo (MCMC) is a natural sampling method that plays a central role in
the computational statistics behind most Bayesian inference techniques \cite%
{brooks2011handbook}.

In our setup, we assume that for any $\bm{\theta}$, one can simulate:
\BEA
\mathbf{y}_{i}(\bm{\theta})=g(\mathbf{x}_{i}(\bm{\theta}),\bm{\xi }_{i}), i=1,\ldots ,T.\label{generalobsmodel}
\EEA
where $\mathbf{x}_{i}(\bm{\theta})\equiv \mathbf{x}(t_{i};\bm{\theta })$ denote
solutions to the initial value problem in Equation (\ref{Eqn:ode}). If the observation function has the following form,
\BEA
g(\mathbf{x}_{i}(\bm{\theta}),\bm{\xi }_{i}) = h(\mathbf{x}_{i}(\bm{\theta})) + \bm{\xi }_{i},\label{additiveobsmodel}
\EEA
where $\bm{\xi }_{i}$ are i.i.d.~noises,
then one can define the likelihood function {of} $\bm{\theta }$, $p(\mathbf{y}%
^{\dag }|\bm{\theta })$, as a product of the density functions of the noises $\bm{\xi }_{i}$,%
\begin{equation}
p(\mathbf{y}^{\dag }|\bm{\theta })\equiv \prod_{i=1}^{T}p(\bm{\xi }%
_{i})=\prod_{i=1}^{T}p(\mathbf{y}^{\dag }_{i}-h(\mathbf{x}_{i}(%
\bm{\theta }))). \label{Eqn:trad_like}
\end{equation}
When the observations are noise-less, $\bm{\xi }_{i}=0$, and the underlying system is an It\^o diffusion process with additive or multiplicative noises, one can use the Bayesian imputation to approximate the likelihood function \cite{golightly2010markov}. In both parametric approaches, it is worth noting that the dependence of the likelihood function on the parameter is implicit through the solutions $\mathbf{x}_{i}(\bm{\theta })$. Practically, this implicit dependence is the source of the computational burden in evaluating the likelihood function since it requires solving the dynamical model in \eqref{Eqn:ode} or every proposal in the MCMC chain. In the case when simulating $\mathbf{y}_i(\bm{\theta})$ is computationally feasible, but the likelihood function is intractable, then one can use, e.g., the Approximate Bayesian Computation (ABC) rejection algorithm \cite{tavare1997inferring,turner2012tutorial} for Bayesian inference. Basically, the ABC rejection scheme generates the samples of $p(\bm{\theta}|\mathbf{y}^\dagger)$ by comparing the simulated $\mathbf{y}_i(\bm{\theta})$ to the observed data, $\mathbf{y}_i^\dagger$, with an appropriate choice of metric comparison for each proposal $\bm{\theta} \sim p_0(\bm{\theta})$. In general, however, repetitive evaluation of \eqref{generalobsmodel} can be expensive when the dynamics in \eqref{Eqn:ode} is high-dimensional and/or stiff, or when $T$ is large, or when the function $g$ is an average of a long time series. Our goal is to address this situation in addition to not knowing the approximate likelihood function.

Broadly speaking, the existing approaches to overcome repetitive evaluation of \eqref{generalobsmodel} require knowledge of an approximate likelihood function such as in \eqref{Eqn:trad_like}. They can be grouped into two classes. 
The first class consists of methods that improve/accelerate the sampling strategy; for example, the Hamiltonian
Monte Carlo \cite{neal2011mcmc}, adaptive MCMC \cite{beck2002Bayesian}, and
delay rejection adaptive Metropolis \cite{haario2006dram}, just to name a
few. The second class consists of methods that avoid solving the dynamical
model in \eqref{Eqn:ode} when running the MCMC chain by replacing it with a
computationally more efficient model on a known parameter domain. This class of approach, also known as
\textit{surrogate modeling}, includes Gaussian process models \cite%
{higdon2004combining}, polynomial chaos \cite{mnr:07,mx:09}, and enhanced
model error \cite{huttunen2007approximation}; for example, the non-intrusive spectral projection \cite{mx:09} approximate $\mathbf{x}_{i}(\bm{\theta})$ in \eqref{Eqn:trad_like} with a polynomial chaos expansion.
Another related approach, which also avoids MCMC on top of integrating \eqref{Eqn:ode}, is to employ a
polynomial expansion on the likelihood function \cite{nagel2016spectral}.
This method represents the parametric likelihood function in %
\eqref{Eqn:trad_like} with orthonormal basis functions of a Hilbert space
weighted by the prior measure. This choice of basis functions makes the
computation for the statistics of the posterior density straightforward, and
thus, MCMC is not needed.

In this paper, we consider a surrogate modeling approach where a nonparametric likelihood function is
constructed using a data-driven spectral expansion. By nonparametric, we mean that our approach does not require any parametric form or assume any distribution as in \eqref{Eqn:trad_like}. Instead, we approximate the likelihood function using the kernel embedding of conditional distribution formulation introduced in \cite{Song2009hilbert,Song2013IEEE}. In {our application}, we will extend their formulation onto a Hilbert space weighted by the sampling measure of the training dataset as in \cite{Berry2017MWR}. We will rigorously {demonstrate} that using orthonormal basis functions of this data-driven weighted Hilbert space, the error bound is independent of the variance of the basis functions, which allows us to determine the amount of training data for accurate likelihood function estimations.

Computationally, assuming that the observations lie on (or close to) a Riemannian manifold $\mathcal{N}$ embedded in $\mathbb{R}^n$ with sampling density $q(\mathbf{y})$, we apply the diffusion maps algorithm \cite{Coifman2006ACHA,Berry2016ACHA} to approximate orthonormal basis functions $\varphi_k\in L^2(\mathcal{N},q)$ using the training dataset. Subsequently, a nonparametric likelihood function is represented as a weighted sum of these data-driven basis functions, where the coefficients are precomputed using the kernel embedding formulation. In this fashion, our approach respects the geometry of the data manifold. Using this nonparametric likelihood function, we then generate the MCMC chain for estimating the conditional distribution of hidden parameters. For the present work, our aim is to demonstrate that one can obtain accurate and robust parameter estimation by implementing a simple Bayesian inference algorithm, the Metropolis scheme, with the data-driven nonparametric likelihood function.
We should also point out that the present method is computationally feasible on low-dimensional parameter space, like any other surrogate modeling approach. Possible ways to overcome this dimensionality issue will be discussed.

This paper is organized as follows: In Section~\ref{sec:rkhs3}, we review the formulation of the reproducing kernel Hilbert space to estimate conditional density functions. In Section~\ref{subsec:err_est}, we
discuss the error estimate of the likelihood function approximation. In Section~%
\ref{sec:basis_func}, we discuss the construction of the analytic basis
functions for the Euclidean data manifold, as well as the data-driven basis
functions with the diffusion maps algorithm for data that lie on embedded Riemannian
geometry. In Section~\ref{sec:numerics}, we provide numerical results with
parameter estimation application on instructive examples. In one of the examples where the dynamical model is low-dimensional and the observation is in the form of \eqref{additiveobsmodel}, we compare the proposed approach with the direct MCMC and non-intrusive spectral projection method (both schemes use likelihood of the form \eqref{Eqn:trad_like}). In addition, we will also demonstrate the robustness of the proposed approach on an example where $g$ is a statistical average of a long-time trajectory (in which the likelihood is intractable) and the dynamical model has relatively high-dimensional chaotic dynamics such that repetitive evaluation of \eqref{generalobsmodel} is numerically expensive. In Section~\ref{sec:conclus}, we conclude this paper with a short summary. We accompany this paper with Appendices for treating large amount of data and more numerical results.

\section{\label{sec:rkhs3}Conditional Density Estimation via Reproducing Kernel Weighted Hilbert Spaces}

Let $\mathbf{y}\in \mathcal{N}\subseteq \mathbb{R}^{n}$,
where $\mathcal{N}$ is a smooth manifold with intrinsic dimension $d\leq n$. In practice, we measure the observations in the ambient coordinates and denote their components as $\mathbf{y}=\{y^{1},\ldots,y^{n}\}$.
For the parameter $\bm{\theta}$ space, $\mathcal{M}$ has a Euclidean structure with components, $\bm{\theta}=\{\theta ^{1},\ldots ,\theta ^{m}\}$, so $\mathcal{M}$ is assumed to be either an $m$-dimensional hyperrectangle or $\mathbb{R}^m$. For training, we are given $M$ number of training parameters $\left\{ %
\bm{\theta }_{j}\right\} _{j=1,\ldots ,M}=\{\theta _{j}^{1},\ldots ,\theta
_{j}^{m}\}_{j=1,\ldots ,M}$. For each training parameter $\bm{\theta }_{j}$,
we generate a discrete time series of length $N$ for noisy observation data $%
\left\{ \mathbf{y}_{i,j}\right\} =\{y_{i,j}^{1},\ldots ,y_{i,j}^{n}\}\in
\mathbb{R}^{n}$ for $i=1,\ldots ,N,$ and $j=1,\ldots ,M$. Here, the
sub-index $i$ and the sub-index $j$\ of $\mathbf{y}_{i,j}$ correspond to the
$i^{\text{th}}$ observation data for the $j^{\text{th}}$ training parameter $\bm{\theta }_{j}$.
Our goal {for training} is to learn the conditional density $p(\mathbf{y}|\bm{\theta }%
) $ from the training dataset
$\left\{ \bm{\theta }_{j}\right\} _{j=1,\ldots ,M}$ and $\left\{ \mathbf{y}%
_{i,j}\right\} _{j=1,\ldots ,M}^{i=1,\ldots ,N}$ for arbitrary $\mathbf{y}$ and $\bm{\theta}$ within the range of
$\left\{ \bm{\theta }_{j}\right\} _{j=1,\ldots ,M}$.

The construction of the conditional density $p(\mathbf{y}|\bm{\theta })$ is
based on a machine learning tool known as the kernel embedding of the
conditional distribution formulation introduced in \cite%
{Song2009hilbert,Song2013IEEE}. In {their} formulation, the representation of
conditional distributions is an element of a Reproducing Kernel Hilbert
Space (RKHS).

Recently, the representation using a Reproducing Kernel Weighted Hilbert Space (RKWHS) was introduced in \cite{Berry2017MWR}. That is, let $\Psi_k:=\psi_k q$ be the orthonormal basis of $L^2(\mathcal{N},q^{-1})$, where they are eigenbasis of an integral operator,
\BEA
\mathcal{K} f(\mathbf{y}) = \int_{\mathcal{N}} K(\mathbf{y},\mathbf{y}') f(\mathbf{y}') q^{-1}(\mathbf{y}')dV, \quad f\in L^2(\mathcal{N},q^{-1}),\label{HSintegral}
\EEA
that is, $\mathcal{K} \Psi_k = \lambda_k\Psi_k$.

In the case where $\mathcal{N}$ is compact and $\mathcal{K}$ is Hilbert--Schmidt, the kernel can be written as,
\BEA
K(\mathbf{y},\mathbf{y}') = \sum_{k=1}^\infty \lambda_k \Psi_k(\mathbf{y}) \Psi_k(\mathbf{y}'),\label{RKHSkernel}
\EEA
which converges in $L^2(\mathcal{N},q^{-1})$. Define the feature map $\Phi: \mathcal{M}\to \ell_2$ as,
\BEA
\Phi(\mathbf{y}) := \big\{\Phi_k(\mathbf{y}) = \sqrt{\lambda}_k \Psi_k(\mathbf{y}) : k\in\mathbb{Z}^+, \mathbf{y}\in\mathcal{N} \big\}.\label{feature}
\EEA
Therefore, any $f\in L^2(\mathcal{N},q^{-1})$ can be represented as $f = \sum_{k=1}^\infty \hat f_k \Psi_k = \sum_{k=1}^\infty \frac{\hat{f_k}}{\sqrt{\lambda_k}} \Phi_k$, where $\hat f_k = \langle f, \Psi_k\rangle_{{q}^{-1}} = \langle f, \psi_k\rangle :=\int_\mathcal{N} f(\mathbf{y}) \psi_k(\mathbf{y}) dV$ and provided that $\sum_k |\hat f_k|^2/\lambda_k <\infty$. If we define $\langle f,g \rangle_{\mathcal{H}_{{q}^{-1}}} := \sum_{k=1}^\infty \frac{\hat{f_k}\hat{g_k}}{\lambda_k}$, we can write the kernel in \eqref{RKHSkernel} as $K(\mathbf{y},\mathbf{y}') = \langle \Phi(\mathbf{y}),\Phi(\mathbf{y'})\rangle_{\mathcal{H}_{{q}^{-1}}}$. Throughout this manuscript, we denote the RKHS $\mathcal{H}_{q^{-1}}(\mathcal{N})$ generating the feature map $\Phi$ in \eqref{feature} as the space of square integrable functions with a reproducing property,
\BEA
f(\mathbf{y}) &=& \langle f, K(\cdot,\mathbf{y})\rangle_{\mathcal{H}_{{q}^{-1}}} := \sum_{k=1}^\infty \frac{\hat{f}_k \langle K(\cdot,\mathbf{y}),\Psi_k \rangle_{q^{-1}} }{\lambda_k}= \sum_{k=1}^\infty \frac{\hat{f}_k}{\sqrt{\lambda_k}} \Phi_k (\mathbf{y}) = \langle f, \Phi(\mathbf{y})\rangle_{\mathcal{H}_{{q}^{-1}}}, \quad\forall \mathbf{y}\in\mathcal{N},\nonumber
\EEA
induced by the basis of $\Psi_k\in L^2(\mathcal{N},q^{-1})$. While this definition deceptively suggests that $\mathcal{H}_{q^{-1}}(\mathcal{N})$ is similar to $L^2(\mathcal{N},q^{-1})$, we should also point out that the RKHS requires that the Dirac functional $\delta_x:\mathcal{H}_{q^{-1}}(\mathcal{N}) \to \mathbb{R}$ defined as $\delta_x f = f(x)$ be continuous. Since $L^2$ contains a class of functions, it is not an RKHS and $\mathcal{H}_{q^{-1}}(\mathcal{N}) \subset L^2(\mathcal{N},q^{-1})$. See, e.g., Chapter~4 of \cite{christmann2008support} for more details. Using the same definition, we denote $\mathcal{H}_{\tilde{q}^{-1}}(\mathcal{M})$ as the RKHS induced by orthonormal basis of $L^2(\mathcal{M},\tilde{q}^{-1})$ of functions of the parameter $\bm{\theta}$.

In this work, we will represent conditional density functions using the RKWHS induced by the data, where the bases will be constructed using the diffusion maps algorithm. The outcome of the training is an estimate of the conditional density, $\widehat{p}(\mathbf{y}|\bm{\theta })$, for arbitrary $\mathbf{y}$
and {$\bm{\theta }$ within the range of $\left\{ \bm{\theta }_{j}\right\} _{j=1,\ldots ,M}$}.

\subsection{\label{subsec:rhks_repre}Review of {Nonparametric} RKWHS Representation of
Conditional Density Functions}

We first review the RKWHS representation of conditional density functions deduced in \cite{Berry2017MWR}. Let $\psi _{k}\left( \mathbf{y}\right)$ be the orthonormal basis functions of $L^{2}\left( \mathcal{N},q\right) $, where $\mathcal{N}$ {contains} the domain of the training data $\mathbf{y}_{i,j}$, and the weight function
$q\left( \mathbf{y}\right) $ is defined with respect to the volume form
inherited by $\mathcal{N}$ from the ambient space $\mathbb{R}%
^{n}$. Let $\varphi _{l}\left( \bm{\theta }\right) \in L^{2}\left( \mathcal{M%
},\widetilde{q}\right) $\ be the orthonormal basis functions in {\ the}
parameter $\bm{\theta }$\ space,
where the training parameters are $%
\bm{\theta }_{j}\in \mathcal{M}$, with weight
function $\widetilde{q}\left( \bm{\theta }\right) $.
For finite modes, $k=1,\ldots ,K_{1}$, and $l=1,\ldots
,K_{2}$, a {nonparametric} RKWHS representation of the conditional density can be written as
follows \cite{Berry2017MWR}:%
\begin{equation}
\widehat{p}\left( \mathbf{y}|\bm{\theta }\right) =\sum_{k=1}^{K_{1}}\widehat{%
c}_{\mathbf{Y}|\bm{\theta },k}\psi _{k}\left( \mathbf{y}\right) q\left(
\mathbf{y}\right) , \label{Eqn:pyb_repre}
\end{equation}%
where $\widehat{p}\left( \mathbf{y}|\bm{\theta }\right) $ denotes an
estimate of the conditional density $p\left( \mathbf{y}|\bm{\theta }\right)\in \mathcal{H}_{q^{-1}}(\mathcal{N})$, and the expansion coefficients are defined as:
\begin{equation}
\widehat{c}_{\mathbf{Y}|\bm{\theta },k}=\sum_{l=1}^{K_{2}}\left[ \mathbf{C}_{%
\mathbf{Y}\bm{\Theta }}\mathbf{C}_{\bm{\Theta \Theta }}^{-1}\right]
_{kl}\varphi _{l}\left( \bm{\theta }\right) . \label{Eqn:ck_repre}
\end{equation}%
Here, the matrix $\mathbf{C}_{\mathbf{Y}\bm{\Theta }}$ is $K_{1}\times K_{2}$, and the
matrix $\mathbf{C}_{\bm{\Theta}\bm{\Theta }}$ is $K_{2}\times K_{2}$, whose
components can be approximated by Monte Carlo averages \cite{Berry2017MWR}:%
\begin{eqnarray}
\left[ \mathbf{C}_{\mathbf{Y}\bm{\Theta }}\right] _{ks} &=&\mathbb{E}_{\mathbf{%
Y}\bm{\Theta }}\left[ \psi _{k}\varphi _{s}\right] \approx \frac{1}{MN}%
\sum_{j=1}^{M}\sum_{i=1}^{N}\psi _{k}\left( \mathbf{y}_{i,j}\right) \varphi
_{s}\left( \bm{\theta }_{j}\right) , \label{Eqn:Cyb} \\
\left[ \mathbf{C}_{\bm{\Theta}\bm{\Theta }}\right] _{sl} &=&\mathbb{E}_{%
\bm{\Theta \Theta }}\left[ \varphi _{s}\varphi _{l}\right] \approx \frac{%
1}{M}\sum_{j=1}^{M}\varphi _{s}\left( \bm{\theta }_{j}\right) \varphi
_{l}\left( \bm{\theta }_{j}\right) , \label{Eqn:Cbb}
\end{eqnarray}%
where the expectations $\mathbb{E}$\ are taken with respect to the sampling
densities of the training dataset $\left\{ \mathbf{y}_{i,j}\right\}
_{j=1,\ldots ,M}^{i=1,\ldots ,N}$ and $\left\{ \bm{\theta }_{j}\right\}
_{j=1,\ldots ,M}$.
The equation for the
expansion coefficients in Eq. (\ref{Eqn:ck_repre})\ is based on the theory
of kernel embedding of the conditional distribution \cite%
{Song2009hilbert,Song2013IEEE,Berry2017MWR}. See \cite{Berry2017MWR}
for the detailed proof of Equations (\ref{Eqn:ck_repre})--(\ref{Eqn:Cbb}). Note
that for RKWHS representation, the weight functions $q$ and $\widetilde{q}$
can be different from the sampling densities of the training dataset $%
\left\{ \mathbf{y}_{i,j}\right\} _{j=1,\ldots ,M}^{i=1,\ldots ,N}$ and $%
\left\{ \bm{\theta }_{j}\right\} _{j=1,\ldots ,M}$, respectively. This
generalizes the representation in \cite{Berry2017MWR}, which sets the
weights $q$ and $\widetilde{q}$ to be the sampling densities of the training
dataset $\left\{ \mathbf{y}_{i,j}\right\} $ and $\left\{ \bm{\theta }%
_{j}\right\} $, respectively. {\color{black}If the assumption of $p\left( \mathbf{y}|\bm{\theta }\right)\in \mathcal{H}_{q^{-1}}(\mathcal{N})$ is not satisfied, then $\mathbf{C}_{\bm{\Theta}\bm{\Theta }}$ can be singular. In such a case, one can follow the suggestion in \cite{Song2009hilbert,Song2013IEEE} to regularize the linear regression in \eqref{Eqn:ck_repre} by replacing $\mathbf{C}_{\bm{\Theta \Theta }}^{-1}$ with $(\mathbf{C}_{\bm{\Theta \Theta }}+\lambda\bm{I}_{K_2})^{-1}$, where $\lambda\in\mathbb{R}$ is an empirically-chosen parameter and $\bm{I}_{K_2}$ denotes an identity matrix of size $K_2\times K_2$.}

Incidentally, it is worth mentioning that the conditional density in (\ref%
{Eqn:pyb_repre})\ and (\ref{Eqn:ck_repre})\ is represented as a regression
in infinite-dimensional spaces with basis functions $\psi _{k}\left( \mathbf{%
y}\right) $ and $\varphi _{l}\left( \bm{\theta }\right) $. The expression (%
\ref{Eqn:pyb_repre}) is a nonparametric representation in the sense that we
do not assume any particular distribution for the density function $p\left(
\mathbf{y}|\bm{\theta }\right) $. In this representation, only training
dataset $\left\{ \mathbf{y}_{i,j}\right\} _{j=1,\ldots ,M\text{ }%
}^{i=1,\ldots ,N}$ and $\left\{ \bm{\theta }_{j}\right\} _{j=1,\ldots ,M}$ {%
\ with appropriate basis functions are used to specify the coefficients $%
\widehat{c}_{\mathbf{Y}|\bm{\theta },k}$ and {the densities $\widehat{p}\left(
\mathbf{y}|\bm{\theta }\right) $}. In Section~\ref{sec:basis_func}, we
will demonstrate how to construct the appropriate basis completely from the training data, motivated by the theoretical result in Section~\ref{subsec:err_est} below.%
}

\subsection{\label{subsec:rhks_simple}Simplification of {the} Expansion
Coefficients (\protect\ref{Eqn:ck_repre})}

If the weight function $\widetilde{q}\left( \bm{\theta }\right) $ is the
sampling density of the training parameters $\left\{ \bm{\theta }%
_{j}\right\} _{j=1,\ldots ,M}$, the matrix $\mathbf{C}_{\bm{\Theta
\Theta }}$ in (\ref{Eqn:Cbb})\ can be simplified to a $K_{2}\times K_{2}$
identity matrix,%
\begin{equation}
\left[ \mathbf{C}_{\bm{\Theta \Theta }}\right] _{sl}=\mathbb{E}_{\bm{%
\Theta \Theta }}\left[ \varphi _{s}\varphi _{l}\right] =\int_{\mathcal{M}%
}\varphi _{s}(\bm{\theta }) \varphi _{l}(\bm{\theta }) \widetilde{q}(%
\bm{\theta }) d \bm{\theta } =\delta _{sl}.
\label{Eqn:Ctheta_delta}
\end{equation}%
where $\delta _{sl}$\ is the Kronecker delta
function. Here, the second equality follows from the weight $\widetilde{q}%
\left( \bm{\theta }\right) $ being the sampling density, and the third
equality follows from the orthonormality of $\varphi _{l}\left( \bm{\theta }%
\right) \in L^{2}\left( \mathcal{M},\widetilde{q}\right) $ with respect to
the weight function $\widetilde{q}$. Then, the expansion coefficients $%
\widehat{c}_{\mathbf{Y}|\bm{\theta },k}$\ in (\ref{Eqn:ck_repre}) can be
simplified to,
\begin{equation}
\widehat{c}_{\mathbf{Y}|\bm{\theta },k}=\sum_{l=1}^{K_{2}}\left[ \mathbf{C}_{%
\mathbf{Y}\bm{\Theta }}\right] _{kl}\varphi _{l}\left( \bm{\theta }\right) ,
\label{Eqn:coeff_simp}
\end{equation}%
with the $K_{1}\times K_{2}$\ matrix $\mathbf{C}_{\mathbf{Y}\bm{\Theta }}$
still given by (\ref{Eqn:Cyb}). In this work, we always take the weight
function $\widetilde{q}\left( \bm{\theta }\right) $ to be the sampling
density of the training parameters $\left\{ \bm{\theta }_{j}\right\}
_{j=1,\ldots ,M}$ for the simplification of the expansion coefficients $%
\widehat{c}_{\mathbf{Y}|\bm{\theta },k}$\ in (\ref{Eqn:coeff_simp}).
This assumption is not too restrictive since the training parameters are specified by the users.

Finally, the formula in (\ref{Eqn:pyb_repre})\ combined with the expansion
coefficients $\widehat{c}_{\mathbf{Y}|\bm{\theta },k}$\ in (\ref%
{Eqn:coeff_simp}) and the matrix $\mathbf{C}_{\mathbf{Y}\bm{\Theta }}$\ in (\ref%
{Eqn:Cyb}) forms an RKWHS representation of the conditional density $p\left(
\mathbf{y}|\bm{\theta }\right) $ for arbitrary $\mathbf{y}$ and $\bm{\theta }$.
Numerically, the training outcome is the matrix $\mathbf{C}_{\mathbf{%
Y}\bm{\Theta }}$ in (\ref{Eqn:Cyb}), and then, the conditional density $\widehat{p}%
\left( \mathbf{y}|\bm{\theta }\right) $ can be represented by (\ref%
{Eqn:pyb_repre}) with coefficients (\ref{Eqn:coeff_simp}) using the basis functions $%
\left\{ \psi _{k}\left( \mathbf{y}\right) \right\} _{k=1}^{K_{1}}$ and $%
\left\{ \varphi _{l}\left( \bm{\theta }\right) \right\} _{l=1}^{K_{2}}$.
From above, one can see that two important questions naturally arise as a
consequence of the usage of RKWHS representation: first, whether the
representation $\widehat{p}%
\left( \mathbf{y}|\bm{\theta }\right) $ in (\ref%
{Eqn:pyb_repre})\ is valid in estimating the conditional density $p\left(
\mathbf{y}|\bm{\theta }\right) $; second, how to construct the orthonormal
basis functions $\psi _{k}\left( \mathbf{y}\right) \in L^{2}\left( \mathcal{N%
},q\right) $ and $\varphi _{l}\left( \bm{\theta }\right) \in L^{2}\left(
\mathcal{M},\widetilde{q}\right) $. We will address these two important
questions {in the next two sections}. 

\section{\label{subsec:err_est}Error estimation}

In this section, we focus on the error estimation of the expansion
coefficient $\widehat{c}_{\mathbf{Y}|\bm{\theta }_{j},k}$ {and, later,} the
conditional density $\widehat{p}\left( \mathbf{y}|\bm{\theta }_{j}\right) $
at the training parameter $\bm{\theta }_{j}$. {The notation
$\widehat{c}_{\mathbf{Y}|\bm{\theta }_{j},k}$ is defined as the expansion
coefficient $\widehat{c}_{\mathbf{Y}|\bm{\theta },k}$\ in (\ref%
{Eqn:coeff_simp}), evaluated at the training parameter $\bm{\theta }_{j} $.} {\
Let the total number of basis functions in parameter space, $K_{2}$, be
equal to the total number of training parameters, $M$, that is,} $K_{2}=M$. 
{\color{black} Denoting $\bm{\Phi} = [\vec{\varphi}_1,\ldots,\vec{\varphi}_M]\in\mathbb{R}^{M\times M}$, where the 
$j^{\text{th}}$ component of $\vec{\varphi}_l$ approximates the basis function evaluated at the training data $\varphi_l(\bm{\theta}_j)$,
we can write the last equality in \eqref{Eqn:Ctheta_delta} in a compact form as ${M}^{-1}\bm{\Phi}^\top\bm{\Phi}=\bm{I}_{M}$. This also means that,
$M^{-1}\bm{\Phi}\bm{\Phi}^\top=\bm{I}_{M}$, the components of which are,}
\begin{equation}
\frac{1}{M}\sum_{l=1}^{M}\varphi _{l}\left( \bm{\theta }_{s}\right) \varphi
_{l}\left( \bm{\theta }_{j}\right) =\delta _{sj}. \label{Eqn:trans_orth}
\end{equation}%
For the training parameter $\bm{\theta }_{j}$, we can {simplify} the expansion
coefficient $\widehat{c}_{\mathbf{Y}|\bm{\theta }_{j},k}$ by substituting
Equation (\ref{Eqn:Cyb})\ into Equation (\ref{Eqn:coeff_simp}),
\begin{eqnarray}
\widehat{c}_{\mathbf{Y}|\bm{\theta }_{j},k} =\sum_{l=1}^{M}\left[ \mathbf{C%
}_{\mathbf{Y}\bm{\Theta }}\right] _{kl}\varphi _{l}\left( \bm{\theta }_{j}\right) 
\approx  \sum_{l=1}^{M}\left[ \frac{1}{MN}\sum_{s=1}^{M}\sum_{i=1}^{N}\psi
_{k}\left( \mathbf{y}_{i,s}\right) \varphi _{l}\left( \bm{\theta }%
_{s}\right) \right] \varphi _{l}\left( \bm{\theta }_{j}\right)  =\frac{1}{N}%
\sum_{i=1}^{N}\psi _{k}\left( \mathbf{y}_{i,j}\right) ,
\label{Eqn:coeffk_train}
\end{eqnarray}%
where the last equality follows from (\ref{Eqn:trans_orth}%
).

\subsection{Error estimation using arbitrary bases}

We first study the error estimation for the expansion coefficient
$\widehat{c}_{\mathbf{Y}|\bm{\theta }_{j},k}$. For each training parameter $%
\bm{\theta }_{j}$, the conditional density function $p(\mathbf{y}|%
\bm{\theta
}_{j})\in \mathcal{H}_{q^{-1}}\left( \mathcal{N}\right)$\ can be analytically
represented in the form,%
\begin{equation}
p(\mathbf{y}|\bm{\theta }_{j})=\sum_{k=1}^{\infty }c_{\mathbf{Y}|\bm{\theta }%
_{j},k}\psi _{k}(\mathbf{y})q(\mathbf{y}), \label{Eqn:com_repres}
\end{equation}%
due to the completeness of $L^{2}\left( \mathcal{N},q\right) $. Here, the
analytic expansion coefficient $c_{\mathbf{Y}|\bm{\theta }_{j},k}$\ is given
by,
\begin{equation}
c_{\mathbf{Y}|\bm{\theta }_{j},k}=\left\langle p\left( \cdot |\bm{\theta }%
_{j}\right) ,\psi _{k} \right\rangle . \label{Eqn:ck_analy}
\end{equation}%
Note that the estimator $\widehat{c}_{\mathbf{Y}|%
\bm{\theta }_{j},k}$ in (\ref{Eqn:coeffk_train}) is a Monte Carlo
approximation of the expansion coefficient $c_{\mathbf{Y}|\bm{\theta }%
_{j},k} $ in (\ref{Eqn:ck_analy}), i.e.,
\begin{equation}
c_{\mathbf{Y}|\bm{\theta }_{j},k}=\left\langle p\left( \cdot|\bm{\theta
}_{j}\right) ,\psi _{k} \right\rangle =\mathbb{E}_{%
\mathbf{Y}|\bm{\theta }_{j}}[\psi _{k}\left( \mathbf{Y}\right) ]\approx
\frac{1}{N}\sum_{i=1}^{N}\psi _{k}\left( \mathbf{y}_{i,j}\right) ,
\label{Eqn:understand}
\end{equation}%
where the last equality follows from the training dataset $\left\{ \mathbf{y}%
_{i,j}\right\} _{i=1,\ldots ,N}$, which admits a conditional density $p\left(
\mathbf{y}|\bm{\theta }_{j}\right) $.
{Note also that in the following theorems and propositions, the condition $p(\mathbf{y}|\bm{\theta
}_{j})\in \mathcal{H}_{q^{-1}}\left( \mathcal{N}\right)$ is required.
In Section \ref{sec:torus3d1p} and Appendix~\ref{sec:dsmalN_KDE}, we will provide an example to discuss this condition in
detail.}
Next, we provide the
unbiasedness and consistency of the estimator $\widehat{c}_{\mathbf{Y}|%
\bm{\theta }_{j},k} $.

\begin{proposition}
\label{thm:cons} Let $\left\{ \mathbf{y}_{i,j}\right\} _{i=1,\ldots ,N}$ be
i.i.d.~samples of $\mathbf{Y}|\bm{\theta }_{j}$ with density $p(\mathbf{y}|%
\bm{\theta }_{j})$. Let $p(\mathbf{y}|\bm{\theta }_{j})\in \mathcal{H}_{q^{-1}}\left( \mathcal{N}\right)$ and $\left\{ \psi _{k}(\mathbf{y})\right\} $
form a complete orthonormal basis of $L^{2}\left( \mathcal{N},q\right) $.
Assume that Var$_{\mathbf{Y}|\bm{\theta }_{j}}\left[ \psi _{k}\left( \mathbf{%
Y}\right) \right] $ is finite, then $\widehat{c}_{\mathbf{Y}|\bm{\theta }%
_{j},k}$ defined in (\ref{Eqn:coeffk_train}) is an unbiased and consistent
estimator for $c_{\mathbf{Y}|\bm{\theta }_{j},k}$ in (\ref{Eqn:ck_analy}).
\end{proposition}

\begin{proof}
The estimator $\widehat{c}_{\mathbf{Y}|\bm{%
\theta }_{j},k}$ is unbiased,%
\begin{equation}
\mathbb{E}\widehat{c}_{\mathbf{Y}|\bm{\theta }_{j},k}=\frac{1}{N}%
\sum_{i=1}^{N}\mathbb{E}_{\mathbf{Y}|\bm{\theta }_{j}}\psi _{k}\left(
\mathbf{Y}_{i,j}\right) =c_{\mathbf{Y}|\bm{\theta }_{j},k}.
\label{Eqn:unbiased}
\end{equation}%
where the expectation is taken with respect to the conditional density $%
p\left( \mathbf{y}|\bm{\theta }_{j}\right) $. If the variance, Var$_{\mathbf{%
Y}|\bm{\theta }_{j}}\left[ \psi _{k}\left( \mathbf{Y}\right) \right] $, is
finite, then the variance of $\widehat{c}_{\mathbf{Y}|\bm{\theta }_{j},k} $
converges to zero as the number of training data $N\rightarrow \infty $,%
\begin{equation}
\text{Var}\left( \widehat{c}_{\mathbf{Y}|\bm{\theta }_{j},k}\right) =\frac{1%
}{N}\text{Var}_{\mathbf{Y}|\bm{\theta }_{j}}\left[ \psi _{k}\left( \mathbf{Y}%
\right) \right] \rightarrow 0,\text{ \ as }N\rightarrow \infty .
\label{Eqn:var_conv}
\end{equation}%
Then, we can obtain that the estimator $\widehat{c}_{\mathbf{Y}|\bm{%
\theta }_{j},k}$ is consistent,
\begin{equation*}
P_{r}\left( \left\vert \widehat{c}_{\mathbf{Y}|\bm{\theta }_{j},k}-c_{%
\mathbf{Y}|\bm{\theta }_{j},k}\right\vert >\varepsilon \right) \leq \frac{%
\text{Var}\left( \widehat{c}_{\mathbf{Y}|\bm{\theta }_{j},k}\right) }{%
\varepsilon ^{2}}\rightarrow 0\text{, as }N\rightarrow \infty \text{, for }%
\forall \varepsilon >0,
\end{equation*}%
where Chebyshev's inequality has been used.
\end{proof}

If the estimator of $p(\mathbf{y}|\bm{\theta }_{j})$ is given by the
representation with an infinite number of basis functions, $\widetilde{p}(%
\mathbf{y}|\bm{\theta }_{j})=\sum_{k=1}^{\infty }\widehat{c}_{\mathbf{Y}|%
\bm{\theta }_{j},k}\psi _{k}(\mathbf{y})q(\mathbf{y})$, then the estimator $%
\widetilde{p}(\mathbf{y}|\bm{\theta }_{j})$ is pointwise unbiased for every
observation $\mathbf{y}$. However, in the numerical implementation, only a finite
number of basis functions can be used in the representation (\ref{Eqn:pyb_repre}%
). Numerically, the estimator of $p(\mathbf{y}|\bm{\theta }_{j})$ is given
by the representation (\ref{Eqn:pyb_repre}) at the training parameter $%
\bm{\theta }_{j}$,%
\begin{equation*}
\widehat{p}(\mathbf{y}|\bm{\theta }_{j})=\sum_{k=1}^{K_{1}}\widehat{c}_{%
\mathbf{Y}|\bm{\theta }_{j},k}\psi _{k}(\mathbf{y})q(\mathbf{y}).
\end{equation*}%
Then, the pointwise error of the estimator, $\widehat{e}(\mathbf{y}|%
\bm{\theta }_{j})$, can be defined as:
\begin{eqnarray}
\widehat{e}(\mathbf{y}|\bm{\theta }_{j}) &\equiv &p(\mathbf{y}|\bm{\theta }%
_{j})-\widehat{p}(\mathbf{y}|\bm{\theta }_{j}) \notag \\
&=&\sum_{k=K_{1}+1}^{\infty }c_{\mathbf{Y}|\bm{\theta }_{j},k}\psi _{k}(%
\mathbf{y})q(\mathbf{y})+\sum_{k=1}^{K_{1}}\left( c_{\mathbf{Y}|\bm{\theta }%
_{j},k}-\widehat{c}_{\mathbf{Y}|\bm{\theta }_{j},k}\right) \psi _{k}(\mathbf{%
y})q(\mathbf{y}).\quad \label{Eqn:p_err2}
\end{eqnarray}%
It can be seen that the estimator $\widehat{p}(\mathbf{y}|\bm{\theta }_{j})$
is no longer unbiased or consistent due to the first error term in (\ref%
{Eqn:p_err2}) induced by modes $k>K_{1}$. Next, we estimate the expectation and the variance of an $L^{2}$-norm error
of $\widehat{p}(\mathbf{y}|\bm{\theta }_{j})$ for all training parameters $%
\bm{\theta }_{j}$.

\begin{theorem}
\label{thm:error1} Let the condition in Proposition \ref{thm:cons} be
satisfied for all $\left\{ \bm{\theta }_{j}\right\} _{j=1,\ldots ,M}$, and
Var$_{\mathbf{Y}|\bm{\theta }_{j}}\left[ \psi _{k}\left( \mathbf{Y}\right) %
\right] $ be finite for all $k\in \mathbb{N}^{+}$. Define the $L^{2}$-norm
error,%
\begin{equation}
\widehat{e}_{L^{2}}=\left( \sum_{j=1}^{M}\int_{\mathcal{N}}\left\vert
\widehat{e}(\mathbf{y}|\bm{\theta }_{j})\right\vert ^{2}q^{-1}(\mathbf{y}%
)dV \right) ^{1/2}, \label{Eqn:L2_err}
\end{equation}%
where $\widehat{e}(\mathbf{y}|\bm{\theta }_{j})$ is the pointwise error in (%
\ref{Eqn:p_err2}), and $dV$\ is the volume form
inherited by the manifold $\mathcal{N}$ from the ambient space $\mathbb{R}%
^{n}$ \cite{Berry2016ACHA,Berry2017MWR}. Then,
\begin{eqnarray}
\mathbb{E}\left[ \widehat{e}_{L^{2}}\right] &\leq &\left(
\sum_{j=1}^{M}\sum_{k=K_{1}+1}^{\infty }\left[ c_{\mathbf{Y}|\bm{\theta }%
_{j},k}\right] ^{2}+\frac{1}{N}\sum_{j=1}^{M}\sum_{k=1}^{K_{1}}\text{Var}_{%
\mathbf{Y}|\bm{\theta }_{j}}\left[ \psi _{k}(\mathbf{Y})\right] \right) ^{%
\frac{1}{2}}, \label{Eqn:EL2} \\
\text{Var}\left[ \widehat{e}_{L^{2}}\right] &\leq
&\sum_{j=1}^{M}\sum_{k=K_{1}+1}^{\infty }\left[ c_{\mathbf{Y}|\bm{\theta }%
_{j},k}\right] ^{2}+\frac{1}{N}\sum_{j=1}^{M}\sum_{k=1}^{K_{1}}\text{Var}_{%
\mathbf{Y}|\bm{\theta }_{j}}\left[ \psi _{k}(\mathbf{Y})\right],
\label{Eqn:varL2}
\end{eqnarray}%
where $\mathbb{E}$ and Var are defined with respect to the joint distribution of
$p(\mathbf{y}|\bm{\theta }_{j})$ for all $\left\{ \bm{\theta }_{j}\right\}
_{j=1,\ldots ,M}$. Moreover, $\mathbb{E}\left[ \widehat{e}_{L^{2}}\right] $
and Var$\left[ \widehat{e}_{L^{2}}\right] $ converge to zero as $%
K_{1}\rightarrow \infty $ and then $N\rightarrow \infty $, where the
limiting operations of $K_{1}$ and $N$ {\ are not commutative.}
\end{theorem}

\begin{proof}
The expectation of $\widehat{e}_{L^{2}}$ can be
estimated as,%
\begin{eqnarray}
\left( \mathbb{E}\left[ \widehat{e}_{L^{2}}\right] \right) ^{2} &\leq &%
\mathbb{E}\left[ \sum_{j=1}^{M}\int_{\mathcal{N}}\left\vert \widehat{e}(%
\mathbf{y}|\bm{\theta }_{j})\right\vert ^{2}q^{-1}(\mathbf{y})dV \right] \notag \\
&=&\mathbb{E}\Big[ \sum_{j=1}^{M}\int_{\mathcal{N}}\Big(
\sum_{k=K_{1}+1}^{\infty }c_{\mathbf{Y}|\bm{\theta }_{j},k}\psi _{k}(\mathbf{%
y})  \nonumber \\ &&+\sum_{k=1}^{K_{1}}\left( c_{\mathbf{Y}|\bm{\theta }_{j},k}-\widehat{c}_{%
\mathbf{Y}|\bm{\theta }_{j},k}\right) \psi _{k}(\mathbf{y})\Big) ^{2}q(%
\mathbf{y})dV \Big],\quad\quad \label{Eqn:err_Jensen}
\end{eqnarray}%
where the first inequality follows from Jensen's inequality. Here, the
randomness comes from the estimators $\widehat{c}_{\mathbf{Y}|\bm{\theta }%
_{j},k}$. Due to the orthonormality of basis functions, $\psi _{k} \in L^{2}\left( \mathcal{N},q\right) $, the error
estimation in (\ref{Eqn:err_Jensen}) can be simplified as,
\begin{eqnarray}
\left( \mathbb{E}\left[ \widehat{e}_{L^{2}}\right] \right) ^{2} &\leq
&\sum_{j=1}^{M}\sum_{k=K_{1}+1}^{\infty }\left[ c_{\mathbf{Y}|\bm{\theta }%
_{j},k}\right] ^{2}+\sum_{j=1}^{M}\sum_{k=1}^{K_{1}}\mathbb{E}_{\mathbf{Y}|%
\bm{\theta }_{j}}\left[ \left( c_{\mathbf{Y}|\bm{\theta }_{j},k}-\widehat{c}%
_{\mathbf{Y}|\bm{\theta }_{j},k}\right) ^{2}\right] , \notag \\
&=&\sum_{j=1}^{M}\sum_{k=K_{1}+1}^{\infty }\left[ c_{\mathbf{Y}|\bm{\theta }%
_{j},k}\right] ^{2}+\frac{1}{N}\sum_{j=1}^{M}\sum_{k=1}^{K_{1}}\text{Var}_{%
\mathbf{Y}|\bm{\theta }_{j}}\left[ \psi _{k}(\mathbf{Y})\right] ,
\label{Eqn:err_var_bound}
\end{eqnarray}%
where the inequality follows from the linearity of expectation, and the
equality follows from $\mathbb{E}\widehat{c}_{\mathbf{Y}|\bm{\theta }%
_{j},k}=c_{\mathbf{Y}|\bm{\theta }_{j},k}$ in (\ref{Eqn:unbiased})\ and Var$%
\left( \widehat{c}_{\mathbf{Y}|\bm{\theta }_{j},k}\right) =\frac{1}{N}$Var$_{%
\mathbf{Y}|\bm{\theta }_{j}}\left[ \psi _{k}\left( \mathbf{Y}\right) \right]
$ in (\ref{Eqn:var_conv}). In error estimation (\ref%
{Eqn:err_var_bound}), the first term is deterministic, and the second term is
{random}. We have so far proven that the expectation $\mathbb{E}\left[
\widehat{e}_{L^{2}}\right] $ is bounded by (\ref{Eqn:EL2}). Similarly, we
can prove that the variance Var$\left[ \widehat{e}_{L^{2}}\right] $ is
bounded by (\ref{Eqn:varL2}).

Next, we prove that the expectation $\mathbb{E}\left[ \widehat{e}_{L^{2}} \right] $ converges to zero as $K_{1}\rightarrow \infty $ and then $%
N\rightarrow \infty $. Parseval's theorem states that:%
\begin{equation}
\sum_{k=1}^{\infty }\left[ c_{\mathbf{Y}|\bm{\theta }_{j},k}\right]
^{2}=\int_{\mathcal{N}}p(\mathbf{y}|\bm{\theta }_{j})^{2}q^{-1}\left(
\mathbf{y}\right) dV <+\infty ,\text{ \ for all }%
\bm{\theta }_{j}, \label{Eqn:Parseval}
\end{equation}%
where the inequality follows from $p(\mathbf{y}|\bm{\theta }_{j})\in \mathcal{H}_{q^{-1}}(\mathcal{N}) \subset
L^{2}\left( \mathcal{N},q^{-1}\right) $ for all $\bm{\theta }_{j}$. For $%
\forall \varepsilon >0$, there exists an integer $\widetilde{K}_{1}\left( %
\bm{\theta }_{j}\right) $ for $\bm{\theta }_{j}$\ such that:%
\begin{equation}
\sum_{k=\widetilde{K}_{1}\left( \bm{\theta }_{j}\right) }^{\infty }\left[ c_{%
\mathbf{Y}|\bm{\theta }_{j},k}\right] ^{2}<\frac{\varepsilon }{2M}.
\label{Eqn:ck_high_bound}
\end{equation}%
Let:
\begin{equation}
K_{1}=\max \left\{ \widetilde{K}_{1}\left( \bm{\theta }_{1}\right) ,\ldots ,%
\widetilde{K}_{1}\left( \bm{\theta }_{M}\right) \right\} , \label{Eqn:K1_op}
\end{equation}%
then the first term in (\ref{Eqn:err_var_bound}) can be bounded by: $%
\varepsilon /2$,%
\begin{equation}
\sum_{j=1}^{M}\sum_{k=K_{1}+1}^{\infty }\left[ c_{\mathbf{Y}|\bm{\theta }%
_{j},k}\right] ^{2}<\frac{\varepsilon }{2}. \label{Eqn:bound_first}
\end{equation}%
Since the variance Var$_{\mathbf{Y}|\bm{\theta }_{j}}\left[ \psi _{k}(%
\mathbf{Y})\right] $ is assumed to be finite for all $k$ and $j$, there exists a constant $D>0$
such that Var$_{\mathbf{Y}|\bm{\theta }%
_{j}}\left[ \psi _{k}(\mathbf{Y})\right] $ can be bounded above by this constant $%
D $,
\begin{equation}
\text{Var}_{\mathbf{Y}|\bm{\theta }_{j}}\left[ \psi _{k}(\mathbf{Y})\right]
\leq D,\text{ \ \ for all }k=1,\ldots ,K_{1}\text{ and }j=1,\ldots ,M.
\label{Eqn:vara_bound}
\end{equation}%
Then, for $\forall \varepsilon >0$, there exists a sufficiently large number
of training data:%
\begin{equation}
N_{\min }=\frac{2MK_{1}D}{\varepsilon }. \label{Eqn:Cond_N}
\end{equation}%
such that whenever $N>N_{\min }$, then:

\begin{equation}
\frac{1}{N}\sum_{j=1}^{M}\sum_{k=1}^{K_{1}}\text{Var}_{\mathbf{Y}|\bm{%
\theta }_{j}}\left[ \psi _{k}(\mathbf{Y})\right] <\frac{\varepsilon }{2}.
\label{Eqn:err_second_bound}
\end{equation}%
Since $\varepsilon>0$ is arbitrary, by substituting Equation (\ref%
{Eqn:bound_first}) and Equation (\ref{Eqn:err_second_bound}) into the error estimation (\ref%
{Eqn:err_var_bound}), we obtain that $\mathbb{E}\left[
\widehat{e}_{L^{2}}\right] $ converges to zero as $K_{1}\rightarrow \infty $
and then $N\rightarrow \infty $.
{Note that, we first take $K_{1}\rightarrow \infty $ to ensure the first error
term in (\ref%
{Eqn:err_var_bound}) vanishes and then take $N\rightarrow \infty $ to ensure
the second error term in (\ref%
{Eqn:err_var_bound}) vanishes. Thus, the limiting operations of $K_{1}$ and $N$
are not commutative.}
Similarly, we can prove that the variance
Var$\left[ \widehat{e}_{L^{2}}\right] $ converges to zero as $%
K_{1}\rightarrow \infty $ and then $N\rightarrow \infty $.
\end{proof}

Theorem \ref{thm:error1} provides the intuition for specifying the number
of training observation data $N$ to achieve any desired accuracy $\varepsilon>0$
given fixed $M$-parameters and sufficiently large $K_1$. It can be seen
from Theorem \ref{thm:error1} that numerically, the expectation $\mathbb{E}%
\left[ \widehat{e}_{L^{2}}\right] $ in (\ref{Eqn:EL2}) and the variance Var$%
\left[ \widehat{e}_{L^{2}}\right] $\ in (\ref{Eqn:varL2}) can be bounded
within arbitrarily small $\varepsilon $ by choosing sufficiently large $%
K_{1} $ and $N$. Specifically, there are two error terms in Equations (\ref%
{Eqn:EL2}) and (\ref{Eqn:varL2}), the first being deterministic, induced by
modes $k>K_{1}$, and the second {random}, induced by modes $k\leq K_{1}$.
For the deterministic term ($k>K_{1}$), the error can be bounded by $%
\varepsilon /2$ by choosing sufficiently large $K_{1}$ satisfying (\ref%
{Eqn:K1_op}).
In our implementation, the number of basis functions $K_{1}$
is empirically chosen to be large enough in order to make the first error
term in Equations (\ref{Eqn:EL2}) and (\ref{Eqn:varL2}) for $k>K_{1}$ as small as
possible.

For the random term ($k\leq K_{1}$), the error\ can be bounded by $%
\varepsilon /2$ by\ choosing sufficiently large $N$ satisfying $N>N_{\min
}=2MK_{1}D/\varepsilon $ (Equation (\ref{Eqn:Cond_N})). The minimum number of
training data, $N_{\min }$, depends on the upper bound of Var$_{\mathbf{Y}|%
\bm{\theta }_{j}}\left[ \psi _{k}(\mathbf{Y})\right] $, $D$. However, the
upper bound $D$ may not exist for some problems. This means that for some
problems, the assumption for finite Var$_{\mathbf{Y}|\bm{\theta }_{j}}\left[
\psi _{k}(\mathbf{Y})\right] $ in Theorem \ref{thm:error1} may not be
satisfied. Even if the upper bound $D$ exists, it is typically not easy to
evaluate its value given an arbitrary basis $\psi _{k}\in L^{2}\left(
\mathcal{N},q\right) $ since one needs to evaluate Var$_{\mathbf{Y}|%
\bm{\theta }_{j}}\left[ \psi _{k}(\mathbf{Y})\right] $ for all $k=1,\ldots
,K_{1}$ and $j=1,\ldots ,M$. Note that Theorem \ref{thm:error1} holds true
for {\ representing} $\widehat{p}(\mathbf{y}|\bm{\theta }_{j})$ with an
arbitrary basis $\left\{\psi _{k}\right\}\in L^{2}\left( \mathcal{N},q\right) $ as long as
$p(\mathbf{y}|\bm{\theta }_{j})\in \mathcal{H}_{q^{-1}}\left( \mathcal{N}\right) $ for all $\theta _{j}$ and Var$_{\mathbf{Y}|\bm{\theta }_{j}}%
\left[ \psi _{k}(\mathbf{Y})\right] $ is finite for $k\leq K_{1}$ and $%
j=1,\ldots ,M$. {\color{black}Next, we provide several cases in which Var$_{%
\mathbf{Y}|\bm{\theta }_{j}}\left[ \psi _{k}(\mathbf{Y})\right] $ is finite
for $k$ and $j$. }

\begin{remark}
\label{rmk:fin1} If the weighted Hilbert space $L^{2}\left( \mathcal{N}%
,q\right) $\ is defined on a compact manifold $\mathcal{N}$ and has smooth basis functions $\psi_k$, then
Var$_{\mathbf{Y}|\bm{\theta }_{j}}\left[ \psi _{k}\left( \mathbf{Y}\right) %
\right] $ is finite for a fixed $k\in \mathbb{N}^{+}$ and $j=1,\ldots ,M$.
{\normalfont This assertion follows from the fact that continuous functions on a compact manifold are bounded.
The smoothness assumption is not unreasonable in many applications since the orthonormal basis functions are
obtained as solutions of an eigenvalue problem of a self-adjoint second-order elliptic differential operator.
Note that the bound here is not necessarily a uniform bound of $%
\psi _{k}\left( \mathbf{Y}\right) $ for all $k\in \mathbb{N}^{+}$ and $%
j=1,\ldots ,M$. As long as Var$_{\mathbf{Y}|\bm{\theta }_{j}}\left[ \psi
_{k}\left( \mathbf{Y}\right) \right] $ is finite for $k\leq K_{1}$ and $%
j=1,\ldots ,M$, the upper bound $D$ is finite, and then, Theorem \ref%
{thm:error1} holds. }
\end{remark}

\begin{remark}
\label{rmk:finite2} If the manifold $\mathcal{N}$ is a {hyperrectangle} in $\mathbb{%
R}^{n}$\ and the weight $q$ is a uniform distribution on $\mathcal{N}$, then
Var$_{\mathbf{Y}|\bm{\theta }_{j}}\left[ \psi _{k}\left( \mathbf{Y}\right) %
\right] $ is finite for a fixed $k\in \mathbb{N}^{+}$ and $j=1,\ldots ,M$.
{\normalfont This assertion is an immediate consequence of Remark \ref%
{rmk:fin1}.}
\end{remark}

In Theorem \ref{thm:error1}, $N_{\min }$ depends on the upper bound of Var$_{%
\mathbf{Y}|\bm{\theta }_{j}}\left[ \psi _{k}(\mathbf{Y})\right] $, $D$, as
shown in (\ref{Eqn:Cond_N}). In the following, we will specify a Hilbert
space, referred to as a data-driven Hilbert space, so that $N_{\min }$ is
independent of $D$ and is only dependent of $M$, $K_{1}$, and $\varepsilon$.
As a consequence, we can easily determine
how many training data $N$ for bounding the second error
term in Equations (\ref{Eqn:EL2}) and (\ref{Eqn:varL2}).

\subsection{Error estimation using a data-driven Hilbert space}

We now turn to the discussion {of} a specific data-driven Hilbert space $%
L^{2}\left( \mathcal{N},\overline{q}\right) $\ {\ with orthonormal basis
functions $\overline{\psi }_{k}$. Our goal is to specify the weight function
$\overline{q}$ such that the minimum number of training} data, $\overline{N}%
_{\min }$, only depends on $M$, $K_{1}$, and $\varepsilon $. Here, the
overline $\overline{\cdot }$\ corresponds to the specific data-driven
Hilbert space. The second error term in (\ref{Eqn:err_var_bound}) can be
further estimated as,
\begin{eqnarray}
\frac{1}{N}\sum_{j=1}^{M}\sum_{k=1}^{K_{1}}\text{Var}_{\mathbf{Y}|\bm{\theta
}_{j}}\left[ \overline{\psi }_{k}(\mathbf{Y})\right] &\leq& \frac{1}{N}%
\sum_{k=1}^{K_{1}}\sum_{j=1}^{M}\mathbb{E}_{\mathbf{Y}|\bm{\theta }_{j}}%
\left[ \overline{\psi }_{k}^{2}(\mathbf{Y})\right] \nonumber \\ &=&\frac{M}{N}%
\sum_{k=1}^{K_{1}}\int_{\mathcal{N}}\overline{\psi }_{k}^{2}(\mathbf{y}%
)\left( \frac{1}{M}\sum_{j=1}^{M}p(\mathbf{y}|\bm{\theta }_{j})\right)
dV, \quad\quad \label{Eqn:improv_sec}
\end{eqnarray}%
where the basis functions are substituted with the specific $\overline{\psi }%
_{k}$. Notice that $\overline{\psi }_{k}(\mathbf{y})$ are orthonormal basis
functions with respect to the weight $\overline{q}$\ in $L^{2}\left( \mathcal{N},\overline{q}\right) $. One specific choice of
the weight function $\overline{q}\left( \mathbf{y}\right) $ is:
\begin{equation}
\overline{q}\left( \mathbf{y}\right) =\frac{1}{M}\sum_{j=1}^{M}p(\mathbf{y}|%
\bm{\theta }_{j}), \label{Eqn:q_bar_y}
\end{equation}%
where $\overline{q}\left( \mathbf{y}\right) $ has been normalized, i.e.,%
\begin{equation}
\int_{\mathcal{N}}\overline{q}\left( \mathbf{y}\right) dV =\int_{\mathcal{N}}\frac{1}{M}\sum_{j=1}^{M}p(\mathbf{y}|\bm{\theta }%
_{j})dV =1. \label{Eqn:norm_q}
\end{equation}%
{For the data-driven Hilbert space, we always use a normalized weight function $\overline{q}\left( \mathbf{y}\right) $.}
{\ Note that the weight function $\overline{q}\left( \mathbf{y}\right) $
in (\ref{Eqn:q_bar_y}) is a discretization of the marginal density
function of $\mathbf{Y}$ with $\bm{\Theta }$ marginalized out,%
\begin{equation}
\overline{q}\left( \mathbf{y}\right) =\frac{1}{M}\sum_{j=1}^{M}p(\mathbf{y}|%
\bm{\theta }_{j})\approx \int_{\mathcal{M}}p(\mathbf{y}|\bm{\theta })%
\widetilde{q}\left( \bm{\theta }\right) d \bm{\theta } =\int_{%
\mathcal{M}}p(\mathbf{y},\bm{\theta })d \bm{\theta } ,
\label{Eqn:margi}
\end{equation}%
where $p(\mathbf{y},\bm{\theta })$ denotes the joint density of $(\mathbf{Y},%
\bm{\Theta })$. Essentially, the weight function} $\overline{q}\left(
\mathbf{y}\right) $ in \eqref{Eqn:q_bar_y} is the sampling density of all
the training data $\left\{ \mathbf{y}_{i,j}\right\} _{j=1,\ldots
,M}^{i=1,\ldots ,N}$, which motivates us to refer to $L^{2}\left(
\mathcal{N},\overline{q}\right) $ as a data-driven Hilbert space.



Next, we prove that by specifying the data-driven basis functions $\overline{%
\psi }_{k}\in L^{2}\left( \mathcal{N},\overline{q}\right) $, the variance Var%
$_{\mathbf{Y}|\bm{\theta }_{j}}\left[ \overline{\psi} _{k}(\mathbf{Y})\right] $ is finite
for all $k\in \mathbb{N}^{+}$ and $j=1,\ldots ,M$. Subsequently, we
can obtain the minimum number of training data, $\overline{N}_{\min }$,
to only depend on $M$, $K_{1}$, and $\varepsilon $, {\ such that} the
expectation $\mathbb{E}\left[ \widehat{e}_{L^{2}}\right] $ in (\ref{Eqn:EL2}%
) and the variance Var$\left[ \widehat{e}_{L^{2}}\right] $\ in (\ref%
{Eqn:varL2}) are bounded above by any $\varepsilon >0$.

\begin{proposition}
\label{lem:var} Let $\left\{ \mathbf{y}_{i,j}\right\} _{i=1,\ldots ,N}$ be
i.i.d.~samples of $\mathbf{Y}|\bm{\theta }_{j}$ with density $p(\mathbf{y}|%
\bm{\theta }_{j})$. Let $p(\mathbf{y}|\bm{\theta }_{j})\in \mathcal{H}_{q^{-1}}(
\mathcal{N})$ for all $\left\{ \bm{\theta }%
_{j}\right\} _{j=1,\ldots ,M}$ {\ with weight $\overline{q}$\ specified in (%
\ref{Eqn:q_bar_y}), and let }$\left\{ {\overline{\psi }_{k}}\right\} ${\ be
the complete orthonormal basis of $L^{2}\left( \mathcal{N},\overline{q}%
\right) $. }Then, Var$_{\mathbf{Y}|\bm{\theta }_{j}}\left[ \overline{\psi }%
_{k}\left( \mathbf{Y}\right) \right] $ is finite for all $k\in \mathbb{N}^{+}
$ and $j=1,\ldots ,M$. \
\end{proposition}

\begin{proof}
Notice that for all $k\in \mathbb{N}^{+}$, we have:%
\begin{eqnarray}
\frac{1}{M}\sum_{j=1}^{M}\text{Var}_{\mathbf{Y}|\bm{\theta }_{j}}\left[
\overline{\psi }_{k}(\mathbf{Y})\right] &\leq& \frac{1}{M}\sum_{j=1}^{M}%
\mathbb{E}_{\mathbf{Y}|\bm{\theta }_{j}}\left[ \overline{\psi }_{k}^{2}(%
\mathbf{Y})\right] =\int_{\mathcal{N}}\overline{\psi }_{k}^{2}(\mathbf{y})%
\frac{1}{M}\sum_{j=1}^{M}p(\mathbf{y}|\bm{\theta }_{j})dV \nonumber \\ &=&\int_{\mathcal{N}}\overline{\psi }_{k}^{2}(\mathbf{y})\overline{q}\left(
\mathbf{y}\right) dV =1, \label{Eqn:each_bound}
\end{eqnarray}%
where the last equality follows directly from the orthonormality of basis functions $%
\overline{\psi }_{k}(\mathbf{y})\in L^{2}\left( \mathcal{N},\overline{q}%
\right) $. From Equation (\ref{Eqn:each_bound}), we can obtain that for all
$k\in \mathbb{N}^{+}$ and $j=1,\ldots ,M$, the variance Var$_{\mathbf{Y}|%
\bm{\theta }_{j}}\left[ \overline{\psi }_{k}\left( \mathbf{Y}\right) \right]
$ is finite.
\end{proof}

\begin{theorem}
\label{thm:var2} Given the same hypothesis as in Proposition \ref{lem:var},
then:
\begin{eqnarray}
\mathbb{E}\left[ \widehat{e}_{L^{2}}\right] &\leq &\left(
\sum_{j=1}^{M}\sum_{k=K_{1}+1}^{\infty }\left[ c_{\mathbf{Y}|\bm{\theta }%
_{j},k}\right] ^{2}+\frac{MK_{1}}{N}\right) ^{\frac{1}{2}},
\label{Eqn:modi_EL2} \\
\text{Var}\left[ \widehat{e}_{L^{2}}\right] &\leq
&\sum_{j=1}^{M}\sum_{k=K_{1}+1}^{\infty }\left[ c_{\mathbf{Y}|\bm{\theta }%
_{j},k}\right] ^{2}+\frac{MK_{1}}{N}. \label{Eqn:var_opt_bound}
\end{eqnarray}%
where $\widehat{e}_{L^{2}}$ is defined by (\ref{Eqn:L2_err}) and $c_{\mathbf{%
Y}|\bm{\theta }_{j},k}$ is given by (\ref{Eqn:ck_analy}). Moreover, $\mathbb{%
E}\left[ \widehat{e}_{L^{2}}\right] $ and Var$\left[ \widehat{e}_{L^{2}}%
\right] $ converge to zero as $K_{1}\rightarrow \infty $ and then $%
N\rightarrow \infty $, where the {\ limiting} operations of $K_{1}$ and $N$ {%
\ are not commutative.}
\end{theorem}

\begin{proof}
According to Proposition \ref{lem:var}, we have
that the variance Var$_{\mathbf{Y}|\bm{\theta }_{j}}\left[ \overline{\psi }%
_{k}\left( \mathbf{Y}\right) \right] $ is finite for all $k\in \mathbb{N}%
^{+} $ and $j=1,\ldots ,M$.
According to Proposition \ref{thm:cons}, since Var$_{\mathbf{Y}|\bm{\theta }_{j}}\left[ \overline{\psi }%
_{k}\left( \mathbf{Y}\right) \right] $ is finite,
we have that the estimator $\widehat{c}_{\mathbf{Y}|\bm{\theta }_{j},k}$\ is
both unbiased and consistent for $c_{\mathbf{Y}|\bm{\theta }_{j},k}$. All conditions in Theorem \ref%
{thm:error1} are satisfied, so that we can obtain the error estimation of the expectation $\mathbb{%
E}\left[ \widehat{e}_{L^{2}}\right] $ in (\ref{Eqn:EL2}) and the error
estimation of the variance Var$\left[ \widehat{e}_{L^{2}}\right] $\ in (\ref%
{Eqn:varL2}). Moreover, the second error term in $\mathbb{%
E}\left[ \widehat{e}_{L^{2}}\right] $ (\ref{Eqn:EL2}) and Var$\left[ \widehat{e}_{L^{2}}\right]$ (%
\ref{Eqn:varL2}) can be both bounded by Equation (\ref{Eqn:each_bound}), so that we can obtain
our error estimations (\ref{Eqn:modi_EL2}) and (\ref{Eqn:var_opt_bound}).

{\ Choose $K_1$ as in (\ref{Eqn:K1_op}) such that the first term in (\ref%
{Eqn:modi_EL2}) and (\ref{Eqn:var_opt_bound}) is bounded by $\varepsilon/2$.}
The second term $MK_{1}/N$\ in (\ref{Eqn:modi_EL2}) and (\ref%
{Eqn:var_opt_bound}) can be bounded by an arbitrarily small $\varepsilon /2$%
\ if the number of training data $N$ satisfies:%
\begin{equation}
N>\overline{N}_{\min }\equiv 2MK_{1}/\varepsilon . \label{Eqn:best_N}
\end{equation}%
Then, both the expectation $\mathbb{E}\left[ \widehat{e}_{L^{2}}\right] $\
and the variance Var$\left[ \widehat{e}_{L^{2}}\right] $\ can be bounded by $%
\varepsilon$. {\ Since $\varepsilon>0$ is arbitrary, the proof is complete.}
\end{proof}

Recall that by applying arbitrary basis functions to represent $\widehat{p%
}\left( \mathbf{y}|\bm{\theta }\right) $ in (\ref{Eqn:pyb_repre}), it is
typically not easy to evaluate the upper bound $D$ in (\ref{Eqn:vara_bound}),
{which implies that it is} not easy to determine how many observation
data, $N_{\min }$ (Equation (\ref{Eqn:Cond_N})), should be {used} for training.
However, by applying the data-driven basis functions $\overline{\psi }_{k}$ %
to represent $\widehat{p}\left( \mathbf{y}|\bm{\theta }\right) $ in (\ref%
{Eqn:pyb_repre}), the minimum number of training data, $\overline{N}_{\min }$
(Equation (\ref{Eqn:best_N})), becomes independent of $D$, and is only dependent
of $M$, $K_{1}$, and $\varepsilon $, as can be seen from Theorem \ref%
{thm:var2}. To let the error induced by modes $k$ $\leq K_{1}$ be smaller
than a desired $\varepsilon /2$, we can easily determine how many
observation data, $\overline{N}_{\min }$ (Equation (\ref{Eqn:best_N})), should be
used for training. In this sense, the specific data-driven Hilbert
space $L^{2}\left( \mathcal{N},\overline{q}\right) $ with the corresponding
basis functions $\overline{\psi }_{k}$ is a good choice for representing (\ref%
{Eqn:pyb_repre}).

{We have so far theoretically verified the validity of the representation (%
\ref{Eqn:pyb_repre})\ in estimating the conditional density $p(\mathbf{y}|\bm{\theta }_{j})$
(Theorem \ref%
{thm:error1}). In particular, using the data-driven basis $%
\overline{\psi }_{k}\in L^{2}\left( \mathcal{N},\overline{q}%
\right) $, we can easily control the error of conditional density estimation by specifying
the number of training data $N$ (Theorem \ref{thm:var2}).
To summarize,} the training procedures can be outlined as follows:

\hangafter1 \hangindent2.85em \setlength{\parindent}{0.0em} \hypertarget{1A}{%
(\textit{1-A})} Generate the training dataset, including training parameters $%
\left\{ \bm{\theta }_{j}\right\} _{j=1,\ldots ,M}$ and observations $\left\{
\mathbf{y}_{i,j}\right\} _{j=1,\ldots ,M}^{i=1,\ldots ,N}$. The length of
training data $N$ is empirically determined based on the criteria (\ref%
{Eqn:Cond_N}) or (\ref{Eqn:best_N}).

\hangafter1 \hangindent2.85em \setlength{\parindent}{0em} \hypertarget{1B}{(%
\textit{1-B})} Construct the basis functions for parameter $\bm{\theta }$
space and for observation $\mathbf{y}$ space by using the training dataset.
For $\mathbf{y}$ space, we need to
empirically choose the number of basis functions $K_{1}$ to let the error
induced by modes $k>K_{1}$ \ be as small as possible.
{In particular, for the data-driven Hilbert space, we will provide a
detailed discussion on how to estimate the data-driven basis functions of
$L^2(\mathcal{N},\overline{q})$ with the sampling density $\overline{q}$ from the training data in the following Section \ref{sec:basis_func}.
Note that this basis estimation will introduce additional errors beyond the results
in this section, which assumed the data-driven basis functions to be given.}

\hangafter1 \hangindent2.85em \setlength{\parindent}{0em} \hypertarget{1C}{(%
\textit{1-C})} Train the matrix $\mathbf{C}_{\mathbf{Y}\bm{\Theta }} $ in (%
\ref{Eqn:Cyb}) and then estimate the conditional density $\widehat{p}\left(
\mathbf{y}|\bm{\theta }\right) $ by using the nonparametric RKWHS
representation (\ref{Eqn:pyb_repre}) with the expansion coefficients $%
\widehat{c}_{\mathbf{Y}|\bm{\theta },k}$ (\ref{Eqn:coeff_simp}).

\hangafter1 \hangindent2.85em \setlength{\parindent}{0em} \hypertarget{1D}{(%
\textit{1-D})} Finally, for new observations $\mathbf{y}^{\dag }=\{\mathbf{y}%
_{1}^{\dag },\ldots ,\mathbf{y}_{T}^{\dag }\}$, define the likelihood function
as a product of the conditional densities of new observations $\mathbf{y}^{\dag
} $ given any $\bm{\theta }$,%
\begin{equation}
p(\mathbf{y}^{\dag }|\bm{\theta })\equiv \prod_{t=1}^{T}\widehat{p}(\mathbf{y%
}_{t}^{\dag }|\bm{\theta }). \label{Eqn:new_like}
\end{equation}

\noindent Next, we address the second important question for the RKWHS
representation (Procedure \hyperlink{1B}{(\textit{1-B})}): how to construct
basis functions for $\bm{\theta }$ and $\mathbf{y}$. {Especially, we focus on
how to construct the data-driven basis functions for $\mathbf{y}$.}

\section{\label{sec:basis_func}Basis functions}

\setlength{\parindent}{1em}

This section will be organized as follows. In Section \ref{subsec:param}, we discuss
how to employ analytical basis functions for parameter $\bm{\theta }$ and
for observation $\mathbf{y}$ as in the usual polynomial chaos expansion.
In Section \ref{subsec:VBDM}, we discuss how to construct the data-driven basis
functions $\overline{\psi }_{k}\in L^{2}\left( \mathcal{N},%
\overline{q}\right) $\ with $\mathcal{N}$ being the manifold of the training
dataset $\left\{ \mathbf{y}_{i,j}\right\} _{j=1,\ldots ,M}^{i=1,\ldots ,N}$
and the weight $\overline{q}$ by (\ref{Eqn:q_bar_y}) being the sampling
density of $\left\{ \mathbf{y}_{i,j}\right\} _{j=1,\ldots ,M}^{i=1,\ldots
,N} $.

\subsection{\label{subsec:param}Analytic basis functions}

If no prior information about the parameter space other than its domain is known, we
can assume that the training parameters\ are uniformly distributed on the
parameter $\bm{\theta}$ space. In particular, we choose $M $ number of well-sampled
training parameters $\left\{ \bm{\theta }_{j}\right\}
_{j=1,\ldots ,M}=\{\theta _{j}^{1},\ldots ,\theta _{j}^{m}\}_{j=1,\ldots ,M}$
in an $m$-dimensional box $\mathcal{M}\subset \mathbb{R}^{m}$,
\begin{equation}
\mathcal{M}\equiv [\theta _{\min }^{1},\theta _{\max }^{1}]\times \cdots
\times \lbrack \theta _{\min }^{m},\theta _{\max }^{m}],
\label{Eqn:theta_unif}
\end{equation}%
where $\times $ denotes a Cartesian product and the two parameters $\theta _{\min
}^{s}$\ and $\theta _{\max }^{s}$\ are the minimum and maximum values of\
the uniform distribution for the $s^{\text{th}}$ coordinate of $\bm{\theta }$
space. Here, the well-sampled uniform distribution corresponds to a regular
grid, which is a tessellation of $m$-dimensional Euclidean space $\mathbb{R}%
^{m}$\ by congruent parallelotopes. Two parameters $\theta _{\min }^{s}$\
and $\theta _{\max }^{s}$\ are determined by: {\
\begin{eqnarray}
\theta _{\min }^{s} &=&\min_{j=1,\ldots ,M}\left\{ \theta _{j}^{s}\right\}
-\gamma \big(\max_{j=1,\ldots ,M}\left\{ \theta _{j}^{s}\right\} -
\min_{j=1,\ldots ,M}\left\{ \theta _{j}^{s}\right\}\big),
\label{Eqn:theta_minmax} \\
\theta _{\max }^{s} &=&\max_{j=1,\ldots ,M}\left\{ \theta _{j}^{s}\right\}
+\gamma \big(\max_{j=1,\ldots ,M}\left\{ \theta _{j}^{s}\right\} -
\min_{j=1,\ldots ,M}\left\{ \theta _{j}^{s}\right\}\big). \notag
\end{eqnarray}%
For $M$ regularly-spaced grid points $\theta_{j}^{s}$, we set $\gamma = .5{%
(M^{s}-1)}^{-1}$ in all of our numerical examples below, where {$M^{s}$} is the number of training parameters in the $s^{\text{th}}$ coordinate.} For example, see
Figure \ref{Fig0wellsample} for the $2$D well-sampled uniformly-distributed
data $\{(5,5),(6,5),\ldots,(12,12)\}$ (blue circles). In this case, the two-dimensional box $\mathcal{M}$\ is $%
[4.5,12.5]^{2}$ (red square).

\begin{figure*}[ptbh]
\centering \includegraphics[scale=0.6]{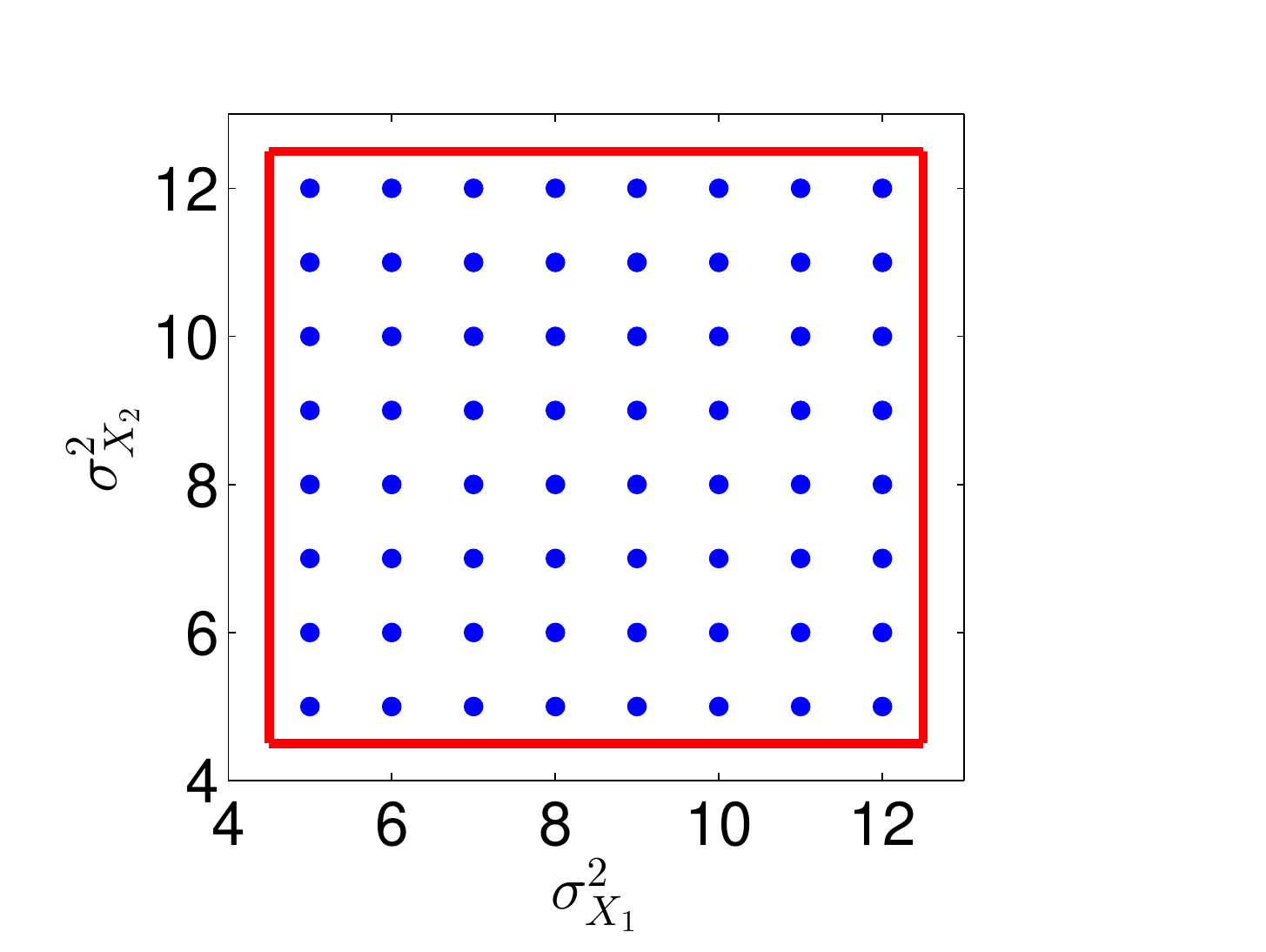}
\caption{(Color online) An example of well-sampled 2D uniformly-distributed
data points (blue circles). The boundary of the uniform distribution is
depicted with a red square. Furthermore, these well-sampled data points correspond
to the training parameters in Example {\hyperref[Ex1_ou]{I} } in Section
\protect\ref{sec:numerics}. In this example, the well-sampled uniformly-distributed training parameters are $(\protect\sigma _{X_{1}}^{2},\protect%
\sigma _{X_{2}}^{2})\in \left\{ (i,j) \right\}_{i=5,\ldots ,12}^{j=5,\ldots ,12}$
(blue circles). The equal
spacing distances of both coordinates are one. The two-dimensional box $%
\mathcal{M}$\ is $[4.5,12.5]^{2}$ (red square). }
\label{Fig0wellsample}
\end{figure*}

{On this simple geometry, we will choose $\varphi _{k}$ to be the tensor
product of the basis functions on each coordinate.} Notice that we have
taken the weight function $\widetilde{q}$ to be the sampling density of the
training parameters in order to simplify the expansion coefficient $\widehat{%
c}_{\mathbf{Y}|\bm{\theta },k}$ in (\ref{Eqn:coeff_simp}). In this case, the
weight $\widetilde{q}$ is a uniform distribution on $\mathcal{M}$. Then, for
the $s^{\text{th}}$ coordinate of the parameter, $\theta ^{s}$, the weight function $%
\widetilde{q}^{s}\left( \theta ^{s}\right) $\ is a uniform distribution on
the interval $[\theta _{\min }^{s},\theta _{\max }^{s}]$, and one can choose
the following cosine basis functions,%
\begin{equation}
\Phi _{k^{s}}(\theta ^{s})=\left\{
\begin{array}{l}
1,\text{ \,\ \ \ \ \ \ \ \ \ \ \ \ \ \ \ \ \ \ \ \ \ \ \ \ \ \ \ \ \ \ \ \ if }%
k^{s}=0\text{,} \\
\sqrt{2}\cos \left( k^{s}\pi \frac{\theta ^{s}-\theta _{\min }^{s}}{\theta
_{\max }^{s}-\theta _{\min }^{s}}\right) ,\text{ else,}%
\end{array}%
\right. \label{Eqn:phi_k}
\end{equation}%
where $\Phi _{k^{s}}(\theta ^{s})$ form a complete orthonormal basis of $%
L^{2}\left( [\theta _{\min }^{s},\theta _{\max }^{s}],\widetilde{q}%
^{s}\right) $. {This choice of basis functions corresponds to exactly the
data-driven basis functions produced by the diffusion maps algorithm on
the uniformly-distributed dataset on a compact interval, which will be
discussed in Section~\ref{subsec:VBDM}.} {Although other choices such as the
Legendre polynomials can be used, this choice will lead to a larger value of
constant $D$ in \eqref{Eqn:Cond_N} that controls the minimum number of training data for accurate estimation. }

Subsequently, we set $L^{2}\left( \mathcal{M},\widetilde{q}\right)
=\bigotimes_{s=1}^{m}L^{2}\left( [\theta _{\min }^{s},\theta _{\max }^{s}],%
\widetilde{q}^{s}\right) $, where $\otimes $ denotes the Hilbert tensor
product, and $\widetilde{q}\left( \bm{\theta }\right) =\bigotimes_{s=1}^{m}%
\widetilde{q}^{s}(\theta ^{s})$ is the uniform distribution on the $m$%
-dimensional box $\mathcal{M}$. Correspondingly, the basis functions $%
\varphi _{k}\left( \bm{\theta }\right) $ are a tensor product of $\Phi
_{k^{s}}(\theta ^{s})$ for $s=1,\ldots ,m$,
\begin{equation}
\varphi _{k}\left( \bm{\theta }\right) =\bigotimes_{s=1}^{m}\Phi
_{k^{s}}(\theta ^{s})=\Phi _{k^{1}}(\theta ^{1})\otimes \ldots \otimes \Phi
_{k^{m}}(\theta ^{m}), \label{Eqn:PHI_K}
\end{equation}%
where $k=\left( k^{1},\ldots ,k^{m}\right) $ and $\bm{\theta }=\left(
\theta ^{1},\ldots ,\theta ^{m}\right) $. Based on the property of the
tensor product of Hilbert spaces, $\left\{ \varphi _{k}\left( \bm{\theta
}\right) \right\} $ forms a complete orthonormal basis of $L^{2}\left(
\mathcal{M},\widetilde{q}\right) $.

We now turn to the discussion of how to construct analytic {\ basis
functions for $\mathbf{y}$. The approach is similar to the one for parameter
$\bm{\theta}$, except that the domain of the data is specified empirically
and the weight function is chosen to correspond to some well-known
analytical basis functions, independent of the sampling distribution of the
data $\mathbf{y}$. That is, we assume the geometry of the data has the
following tensor structure, $\mathcal{N} = \mathcal{N}^1\times \ldots\times
\mathcal{N}^n$, where $\mathcal{N}^s$ will be specified empirically based on
the ambient space coordinate of $\mathbf{y}$. Let $y^{s}$ be the $s^{\text{th%
}}$ ambient component of $\mathbf{y}$; we can choose a weighted Hilbert
space $L^{2}\left( \mathcal{N}^s,q^{s}(y^{s};\alpha ^{s})\right) $ with the
weight $q^{s}$ depending on the parameters $\alpha ^{s}$ and being
normalized to satisfy $\int_{\mathbb{R}}q^{s}(y^{s};\alpha ^{s})dy^{s}=1$.
For each coordinate, let $\Psi _{k^{s}}(y^{s};\alpha ^{s}) $ be the
corresponding orthonormal basis functions, which possess analytic
expressions. Subsequently, we can obtain a set of complete orthonormal basis
functions $\psi _{k}\in L^{2}\left( \mathcal{N},q\right) $\ for $\mathbf{y}$
by taking the tensor product of these $\Psi _{k^{s}}$ as in \eqref{Eqn:PHI_K}%
. }

For example, if the weight $q^{s}$ is uniform, $\mathcal{N}^s \subset
\mathbb{R}$ is simply a one-dimensional interval. In this case, we can
choose the cosine basis functions $\Psi _{k^{s}}$ for $\mathbf{y}$ as in %
\eqref{Eqn:phi_k} such that the parameters $\alpha^{s}$ correspond to the
boundaries of the domain $\mathcal{N}^s$, which can be estimated as in %
\eqref{Eqn:theta_minmax}. In our numerical experiments below, we will set $%
\gamma = 0.1$. Another choice is to set the weight $q^{s}(y^{s};\alpha ^{s})$
to be Gaussian. In this case, the domain is assumed to be the real line, $%
\mathcal{N}^s = \mathbb{R}$. For this choice, the corresponding orthonormal
basis functions $\Psi _{k^{s}}$ are Hermite polynomials, and the parameters $%
\alpha^{s}$, corresponding to the mean and variance of the Gaussian distribution,
can be empirically estimated from the training data.
\comment{,
\begin{equation}
\widehat{\mu }^{s}=\frac{1}{MN}\sum_{j=1}^{M}\sum_{i=1}^{N}y_{i,j}^{s}\text{%
, \ \ }(\widehat{\sigma}^{s})^2=\frac{1}{MN-1}\sum_{j=1}^{M}\sum_{i=1}^{N}%
\left( y_{i,j}^{s}-\widehat{\mu }^{s}\right) ^{2}\text{.}
\label{Eqn:hermite_estimator}
\end{equation}}

{In the remainder of this paper, we will always use the cosine basis
functions for $\bm{\theta}$. The application of (\ref{Eqn:pyb_repre}%
) using cosine basis functions for $\mathbf{y}$ is referred to as the cosine
representation. The application of (\ref{Eqn:pyb_repre}) using Hermite basis
functions for $\mathbf{y}$ is referred to as the Hermite representation.

\subsection{\label{subsec:VBDM}Data-Driven Basis Functions}

In this section, we discuss how to construct a set of data-driven basis
functions $\overline{\psi }_{k}$\ $\in $ $L^{2}\left( \mathcal{N},\overline{q%
}\right) $\ with $\mathcal{N}$ being the manifold of the training dataset $%
\left\{ \mathbf{y}_{i,j}\right\} _{j=1,\ldots ,M}^{i=1,\ldots ,N}$\ and
weight $\overline{q}$ in (\ref{Eqn:q_bar_y}) being the sampling density of $%
\left\{ \mathbf{y}_{i,j}\right\} $ for all $i=1,\ldots ,N,$ and $j=1,\ldots
,M$. The issues here are that the analytical expression of the sampling
density $\overline{q}$ is unknown and the Riemannian metric inherited by the
data manifold $\mathcal{N}$ from the ambient space $\mathbb{R}^n$ is also
unknown. Fortunately, these issues can be overcome by the diffusion maps
algorithm \cite{Coifman2006ACHA,Berry2016ACHA,Berry2017MWR}.

\subsubsection{\label{subsec:vbdm_basis}Learning the Data-Driven Basis
Functions}

Given a dataset $\mathbf{y}_{i,j}\in \mathcal{N}\subseteq \mathbb{R}%
^{n}$ with the sampling density $\overline{q}(\mathbf{y})$ (\ref{Eqn:q_bar_y}%
), defined with respect to the volume form inherited by the manifold $%
\mathcal{N}$ from the ambient space $\mathbb{R}^{n}$, one can use the
kernel-based diffusion maps algorithm to construct an $MN\times MN$ matrix $%
\mathbf{L}$ that approximates a weighted Laplacian operator, $\mathcal{L}%
=\nabla \log \left( \overline{q}\right) \cdot \nabla +\bigtriangleup$, {that takes functions} with Neumann boundary conditions for the compact
manifold $\mathcal{N}$ with the boundary if the manifold has a boundary.
The eigenvectors $\overline{\bm{\psi
}}_{k}$ of the matrix $\mathbf{L}$ are discrete approximations of the
eigenfunctions $\overline{\psi }_{k}\left( \mathbf{y}\right) $ of the
operator $\mathcal{L}$, which form an orthonormal basis of
the weighted Hilbert space $L^{2}\left( \mathcal{N},\overline{q}\right) $.
Connecting to the discussion on the RKWHS in Section~\ref{sec:rkhs3}, the eigenfunctions of $\mathcal{L}^*= -\mbox{div}(\nabla\log\overline{q}\,\, )+\Delta$, that is $\{\Psi_k:= \overline{\psi }_{k}\overline{q}\}$, can be approximated using an integral operator in \eqref{HSintegral} with the appropriate kernel constructed by the diffusion maps algorithm, up to a diagonal conjugation. Basically, $\mathcal{H}_{\bar{q}^{-1}}(\mathcal{N})$ is the data-driven reproducing kernel Hilbert space defined with the feature map in \eqref{feature}, induced by eigenfunctions of $\mathcal{L}^*$.

Each component of the eigenvector $\overline{\bm{\psi }}_{k}\in \mathbb{R}%
^{MN}$ is a discrete estimate of the eigenfunction $\overline{\psi }%
_{k}\left( \mathbf{y}_{i,j}\right) $, evaluated at the training data point $%
\mathbf{y}_{i,j}$. {The sampling density $\overline{q}$ {defined in %
\eqref{Eqn:q_bar_y}} is estimated using a kernel density estimation method \cite{Berry2016forecasting}.}
In contrast to the analytic continuous basis functions in the above Section \ref%
{subsec:param}, the data-driven basis functions $\overline{\psi }_{k} \in L^{2}\left( \mathcal{N},\overline{q}\right) $ are
represented nonparametrically by the discrete eigenvectors $\overline{%
\bm{\psi }}_{k}\in \mathbb{R}^{MN}$ using the diffusion maps algorithm. The
outcome of the training is a discrete estimate of the conditional density, $%
\widehat{p}\left( \mathbf{y}_{i,j}|\bm{\theta }\right) $, which estimates
the representation $\widehat{p}\left( \mathbf{y}|\bm{\theta }\right) $ (\ref%
{Eqn:pyb_repre}) on each training data point $\mathbf{y}_{i,j}$.

In our implementation, we use the Variable-Bandwidth Diffusion Maps (VBDM)
algorithm introduced in \cite{Berry2016ACHA}, which {extends the diffusion maps to non-compact
manifolds without a boundary. See the supplementary material of \cite{harlim2018} for the MATLAB code of this algorithm.
We should point out that this discrete approximation induces errors in the basis function,
which are estimated {\color{black}in detail in \cite{bs:2019}.} These errors are in addition to the error estimations in Section~\ref{subsec:err_est}.}

{We note that if the data are uniformly distributed on a
one-dimensional bounded interval, then the VBDM solutions are the cosine basis
functions, which are eigenfunctions of the Laplacian operator on bounded interval with Neumann
boundary conditions. This means that the cosine functions in %
\eqref{Eqn:phi_k} that are used to represent each component of $\bm{\theta}$
are analogous to the data-driven basis functions. The difference is that with the
parametric choice in \eqref{Eqn:phi_k}, one avoids VBDM at the expense
of specifying the boundaries of the domain, $[\theta^s_{\min},\theta^s_{%
\max}]$. In the remainder of this paper, we refer to an application of (\ref%
{Eqn:pyb_repre}) with cosine basis functions for $\bm{\theta}$ and VBDM
basis functions for $\mathbf{y}$ as the VBDM representation. }

However, a direct application of the VBDM algorithm suffers from the expensive
computational cost for large training data. Basically, we need an algorithm that allows us to
subsample from the training dataset while preserving the sampling distribution of the full dataset.
In Appendix~\ref{subsec:data_clus}, we provide a simple box-averaging method to achieve this goal. In the remainder of this paper, we will denote the reduced data obtained via the box-averaging method in Appendix \ref{subsec:data_clus} by $\{\overline{\mathbf{y}}_b\}_{b=1,\ldots, B}$, where $B \ll MN$. We refer to them as the box-averaged data points. When the number of training data is too large, we apply the VBDM algorithm on these box-averaged data to obtain the discrete estimate of the eigenfunctions $\overline{\psi }%
_{k}\left( \overline{\mathbf{y}}_{b}\right)$.

The second issue arises from the discrete representation of the
conditional density in the observation $\mathbf{y}$ space using the VBDM
algorithm. Notice that the VBDM representation, $%
\widehat{p}\left( \mathbf{y}_{i,j}|\bm{\theta }\right) $, is only estimated
at each training data point $\mathbf{y}_{i,j}$. A natural problem is to
extend the representation onto new observations $\mathbf{y}_{t}\notin
\left\{ \mathbf{y}_{i,j}\right\} _{j=1,\ldots ,M}^{i=1,\ldots ,N}$ that are
not part of the training dataset (Procedure \hyperlink{1D}{(\textit{1-D})}). Next, we address this issue.


\subsubsection{\label{subsec:vbdm_nystrom}Nystr\"{o}m Extension}

We now discuss an extension method to evaluate basis functions $\overline{\psi }_{k}$ on a new data point that does not belong to the training dataset. Given such an extension method, we can proceed with Procedure \hyperlink{1D}{(\textit{1-D})}
by evaluating $\overline{\psi }_{k}\left( \mathbf{y}_{t}\right)$ on new observations $\mathbf{y}_{t}\notin \left\{ \mathbf{y}_{i,j}\right\} _{j=1,\ldots ,M}^{i=1,\ldots ,N}$, which in turn give $\widehat{p}\left( \mathbf{y}_{t}|\bm{\theta }\right)$.
Second, this extension is also needed in the training Procedure \hyperlink{1C}{(\textit{1-C})} when $MN$ is large.
More specifically, for training the matrix $%
\mathbf{C}_{\mathbf{Y}\bm{\Theta }}$ in (\ref{Eqn:Cyb}), we need to know the
estimate of the eigenfunction $\overline{\psi }_{k}\left( \mathbf{y}%
_{i,j}\right) $ for all the original training data
$\mathbf{y}_{i,j}$.
Computationally, however, we can only construct the discrete estimate of the eigenfunction $\overline{\psi }%
_{k}\left( \overline{\mathbf{y}}_{b}\right) $ at the reduced box-averaged data
points $\overline{\mathbf{y}}_{b}$. This suggests that we need to extend the
eigenfunctions $\overline{\psi }_{k}\left( \overline{\mathbf{y}}_{b}\right) $
onto all the original training data $\left\{ \mathbf{y}_{i,j}\right\}
_{j=1,\ldots ,M}^{i=1,\ldots ,N}$. 

For the convenience of discussion, the training data that are used to
construct the eigenfunctions are denoted by $\left\{ \mathbf{y}_{r}^{\text{%
old}}\right\} _{r=1,\ldots ,R}$, and all the data that are not part of
$\left\{ \mathbf{y}_{r}^{\text{old}}\right\} _{r=1,\ldots ,R}$ are denoted
by $\mathbf{y}^{\text{new}}$. To extend the eigenfunctions $\overline{\psi }%
_{k}\left( \mathbf{y}_{r}^{\text{old}}\right) $ onto the data point $\mathbf{%
y}^{\text{new}}\notin \left\{ \mathbf{y}_{r}^{\text{old}}\right\}
_{r=1,\ldots ,R}$, one approach would be to use the Nystr\"{o}m extension
\cite{Nystrom1930Acta} that is based on the basic theory of RKHS \cite%
{Aronszajn1950theory}. Let $\mathcal{H}_{\bar{q}}\left( \mathcal{N}\right)$ be
the RKWHS with a symmetric positive kernel $\widehat{\mathcal{T}}:\mathcal{N}\times \mathcal{N}\mathbb{\rightarrow R}$ defined as,
\BEA
\widehat{\mathcal{T}}(\mathbf{y},\mathbf{y}') = \sum_{k=1}^\infty \lambda_k \overline{\psi }_{k}(\mathbf{y})\overline{\psi }_{k}(\mathbf{y}'), \nonumber
\EEA
where $\lambda_k$ is the corresponding eigenvalue of $\mathcal{L}$ associated with eigenfunction $\overline{\psi }_{k}$. Then, for any function $f\in \mathcal{H}_{\overline{q}}\left(
\mathcal{N}\right) $, the Moore--Aronszajn theorem states that
one can evaluate $f$ at $\mathbf{a}\in \mathcal{N}$ with the following inner
product, $f(\mathbf{a})=\left\langle f,\widehat{\mathcal{T}}\left( \mathbf{a}
,\cdot \right) \right\rangle _{\mathcal{H}_{\overline{q}}}$. In our application, this
amounts to evaluating,%
\begin{equation}
\overline{\psi }_{k}\left( \mathbf{y}^{\text{new}}\right) =\frac{1}{R}%
\sum_{r=1}^{R}\mathcal{T}\left( \mathbf{y}^{\text{new}},\mathbf{y}_{r}^{%
\text{old}}\right) \overline{\psi }_{k}\left( \mathbf{y}_{r}^{\text{old}%
}\right) , \label{Eqn:nystrom}
\end{equation}%
where the non-symmetric kernel function $\mathcal{T}:\mathcal{N\times N}%
\mathbb{\rightarrow R}$ (constructed by the diffusion maps algorithm) is related to the symmetric kernel $\widehat{
\mathcal{T}}$ by,
\BEA
\mathcal{T}\left( \mathbf{y}_{i},\mathbf{y}_{j}\right) =%
\overline{q}^{-1/2}\left( \mathbf{y}_{i}\right) \widehat{\mathcal{T}}\left(
\mathbf{y}_{i},\mathbf{y}_{j}\right) \overline{q}^{1/2}\left( \mathbf{y}%
_{j}\right)\nonumber
\EEA
with $\overline{q}\left( \mathbf{y}_{i}\right) $\ being the
sampling density of $\left\{ \mathbf{y}_{r}^{\text{old}}\right\}
_{r=1,\ldots ,R}$ at $\mathbf{y}_{i}$. See the detailed evaluation of the
kernels $\widehat{\mathcal{T}}$ and $\mathcal{T}$\ for\ the Nystr\"{o}m
extension in \cite{Harlim2017diffusion}. After obtaining the estimate
of the eigenfunction $\overline{\psi }_{k}\left( \mathbf{y}^{\text{new}%
}\right) $ using the Nystr\"{o}m extension, we can train the matrix $\mathbf{%
C}_{\mathbf{Y}\bm{\Theta }}$ in (\ref{Eqn:Cyb}) for large $MN$ and then obtain
the representation of the conditional density on arbitrary new observation $%
\mathbf{y}_{t}$, $\widehat{p}\left( \mathbf{y}_{t}|\bm{\theta }\right) $.

{To summarize this section, we have constructed two different sets of basis
functions for $\mathbf{y}$, the analytic basis functions of $%
L^{2}\left( \mathcal{N},q\right) $\, such as the Hermite and cosine basis
functions, {which assume that the manifold is $\mathbb{R}^n$ or hyperrectangle, respectively,
and} the data-driven basis functions of $L^{2}\left( \mathcal{N},%
\overline{q}\right), $\ with $\mathcal{N}$ being the data manifold and $%
\overline{q}$ being the sampling density that are computed using the VBDM
algorithm.

\section{\label{sec:numerics}Parameter Estimation Using the Metropolis Scheme}

First, we briefly review the Metropolis scheme for estimating the posterior
density $p(\bm{\theta }|\mathbf{y}^{\dag })$ given new observations $\mathbf{%
y}^{\dag }=\{\mathbf{y}_{1}^{\dag },\ldots ,\mathbf{y}_{T}^{\dag }\}$ for a
specific parameter $\bm{\theta }^{\dag }$. The key idea of the Metropolis
scheme is to construct a Markov chain such that it converges to samples of
conditional density $p(\bm{\theta }|\mathbf{y}^{\dag })$ as the target
density. In our application, the parameter estimation procedures can be
outlined as follows:

\hangafter1 \hangindent2.85em \setlength{\parindent}{0em} \hypertarget{2A}{(%
\textit{2-A})} Suppose we have $\bm{\theta }_{0}\sim $ $p(\bm{\theta
}_{0}|\mathbf{y}^{\dag })>0$, then for $i\geq 1$, we can sample $\bm{\theta }%
^{\ast }\sim \kappa \left( \bm{\theta }_{i-1},\bm{\theta }^{\ast }\right) $.
Here, $\kappa $ is the proposal kernel density. For example, use the random
walk Metropolis algorithm to generate proposals, $\kappa \left( \bm{\theta }_{i-1},\bm{\theta }^{\ast }\right) = \mathcal{N}\left(\bm{\theta }_{i-1} ,\mathbf{C}\right)$,
where $\mathbf{C}$, the proposal covariance, and is a tunable nuisance
parameter.

\hangafter1 \hangindent2.85em \setlength{\parindent}{0em} \hypertarget{2B}{(%
\textit{2-B})} Accept the proposal, $\bm{\theta }_{i}=\bm{\theta }^{\ast }$
with probability $\min (\frac{p(\bm{\theta }^{\ast
}|\mathbf{y}^{\dag })}{p(\bm{\theta }_{i-1}|\mathbf{y}^{\dag })},1)$%
, otherwise set $\bm{\theta }_{i}=\bm{\theta }_{i-1}$.
\comment{
Here,
\begin{equation*}
\alpha (\bm{\theta }_{i-1},\bm{\theta }^{\ast })=\frac{p(\bm{\theta }^{\ast
}|\mathbf{y}^{\dag })}{p(\bm{\theta }_{i-1}|\mathbf{y}^{\dag })}.
\end{equation*}%
}
Repeat Procedures \hyperlink{2A}{(\textit{2-A})} and \hyperlink{2B}{(\textit{%
2-B})} above.
Notice that the
posterior $p(\bm{\theta }|\mathbf{y}^{\dag })$ can be determined from the
prior $p_{0}(\bm{\theta })$ and the likelihood $p(\mathbf{y}^{\dag }|%
\bm{\theta })$\ based on Bayes' theorem (\ref{Eqn:bayes}). The likelihood
function $p(\mathbf{y}^{\dag }|\bm{\theta })$ is defined as a product of
conditional densities of new observations $\mathbf{y}^{\dag }=\{\mathbf{y}%
_{1}^{\dag },\ldots ,\mathbf{y}_{T}^{\dag }\}$\ in (\ref{Eqn:new_like})
(Procedure \hyperlink{1D}{(\textit{1-D})}). The conditional densities\ of
new observations $\mathbf{y}^{\dag }$ given $\bm{\theta }$\ are obtained
from the training Procedure \hyperlink{1C}{(\textit{1-C})}.

\hangafter1 \hangindent2.85em \setlength{\parindent}{0em} \hypertarget{2C}{(%
\textit{2-C})} Generate a sufficiently long chain and use the chain's statistic
as an estimator of the true parameter $\bm{\theta} ^{\dag }$. Take multiple runs
of the chain started at different initial $\bm{\theta }_{0}$, and examine
whether all these runs converge to the same distribution. {The convergence of
all the examples below has been validated using 10 randomly-chosen different initial conditions.}

\setlength{\parindent}{1em}

In the remainder of this section, we present numerical results of the Metropolis scheme using the proposed data-driven likelihood function on various instructive examples, where the likelihood function is either explicitly known, or can be approximated as in \eqref{Eqn:trad_like}, or is intractable. In an example where the explicit likelihood is known, our goal is to show that the approach numerically converges to the true posterior estimate. In the second example, where the dimension of the data manifold is strictly less than the ambient dimension, we will show that the RKHS framework with the knowledge of the intrinsic geometry is superior. When the intrinsic geometrical information is unknown, the proposed data-driven likelihood function is competitive. In the third example with a low-dimensional dynamic and observation model of the form \eqref{additiveobsmodel}, we compare the proposed approach with standard methods, including the direct MCMC and nonintrusive spectral projection (both use the likelihood function of the form \eqref{Eqn:trad_like}). In our last example, we consider an observation model where the likelihood function is intractable and the cost of evaluating the observation model in \eqref{generalobsmodel} is numerically expensive.

\subsection{\label{sec:ou2d2p} Example I: Two-Dimensional Ornstein--Uhlenbeck Process}
\label{Ex1_ou}

Consider an Ornstein--Uhlenbeck (OU) process as follows:
\begin{equation}
d\mathbf{X} = -\frac{1}{2}\mathbf{X}dt + \bm{\Sigma }^{1/2}d\mathbf{W}_t,\label{Eqn:ou2d}
\end{equation}
where $\mathbf{X}\equiv \left( X_{1},X_{2}\right)^\top$ denotes the state variable, $%
\mathbf{W}_t = (W_{1},W_{2})^\top$\ denotes two-dimensional Wiener processes, and $\bm{\Sigma }\in\mathbb{R}^{2\times 2}$ is a diagonal matrix with {main diagonal} components $\sigma _{X_{1}}^{2}$ and $\sigma _{X_{s}}^{2}$ to be estimated. In the stationary process
, the solution of Equation (\ref{Eqn:ou2d}) $\mathbf{X}=\left(
X_{1},X_{2}\right)^\top$ admits a Gaussian distribution $\mathbf{X}\sim
\mathcal{N}\left( \mathbf{0},\bm{\Sigma }\right) $,
\begin{equation}
p(\mathbf{X}|\bm{\Sigma })=\det \left( 2\pi \bm{\Sigma }\right) ^{-%
\frac{1}{2}}\exp \left[ -\frac{1}{2}\mathbf{X}^{\top}\bm{\Sigma }^{-1}%
\mathbf{X}\right] . \label{Eqn:station_f2}
\end{equation}
Our goal here is to estimate the posterior density and the posterior mean of
the parameters $(\sigma _{X_{1}}^{2},\sigma _{X_{2}}^{2})$, given a finite number, $T$, of observations, $\mathbf{X}^{\dag
}\equiv (\mathbf{X}^{1\dag },\ldots ,\mathbf{X}^{T\dag })$, for hidden true
parameters $((\sigma _{X_{1}}^{2})^\dag,(\sigma _{X_{2}}^{2})^\dag)=(6.5,6.3)$, {\color{black}where each $\mathbf{X}^{t\dag }$ is an i.i.d. sample of \eqref{Eqn:station_f2} for $\bm{\Sigma }=\bm{\Sigma }^\dagger$.} This example is shown here to verify the validity of the framework of our RKWHS representations for parameter estimation application.

One can show that the likelihood
function for this problem is the inverse matrix gamma distribution, $\bm{\Sigma }\sim IMG\left( \frac{T}{2}-\frac{3}{2},2,\bm{\Psi}
\right),$ where $\bm{\Psi } = \mathbf{X}^{\dag}(\mathbf{X}^{\dag})^\top
\in\mathbb{R}^{2\times 2}$.
If a prior is defined to be also the inverse matrix gamma distribution, $%
\bm{\Sigma }\sim IMG\left( \alpha _{0},2,\mathbf{0}\right) $,%
\comment{
\begin{equation}
p_{0}(\bm{\Sigma })\propto \left\vert \bm{\Sigma }\right\vert
^{-\alpha _{0}-3/2}, \label{Eqn:prior}
\end{equation}%
}
for some value of $\alpha _{0}$, then the posterior density $p(\bm{\Sigma }|\mathbf{X}^{\dag })$ can be obtained
by applying Bayes' theorem,%
\begin{equation}
p(\bm{\Sigma }|\mathbf{X}^{\dag }) \sim IMG\left( \alpha _{0}+\frac{T%
}{2},2,\bm{\Psi }\right) . \label{Eqn:posterior}
\end{equation}
The posterior mean can thereafter be obtained as,
\begin{equation}
\left( \bm{\Sigma }\right) _{PM}=\left(
\begin{array}{cc}
\left( \sigma _{X_{1}}^{2}\right) _{PM} & 0 \\
0 & \left( \sigma _{X_{2}}^{2}\right) _{PM}%
\end{array}%
\right) =\frac{\bm{\Psi }}{T+2\alpha _{0}-3}. \label{Eqn:postmean2}
\end{equation}

To compare with the analytic conditional density $p(\mathbf{X}|\bm{%
\Sigma })$ (\ref{Eqn:station_f2}), we trained three RKWHS representations of
the conditional density function, $\widehat{p}\left( \mathbf{X}|\bm{%
\Sigma }\right) $, by using the same training dataset. For
training, we used $M=64$ well-sampled uniformly-distributed training
parameters (shown in Figure \ref{Fig0wellsample}), $(\sigma
_{X_{1}}^{2},\sigma _{X_{2}}^{2})$, where $\sigma _{X_{j}}^{2} \in\{ 5,6,\ldots,12 \}$, which are denoted
by $\left\{ \bm{\Sigma }_{j}\right\} _{j=1}^{M}$. For each training
parameter $\bm{\Sigma }_{j}$, we generated $N=$ 640,000 well-sampled normally distributed observation data {\color{black}of density in \eqref{Eqn:station_f2} with $\bm{\Sigma}=\bm{\Sigma}_j$}.
For Hermite and cosine representations, we used $20$ basis functions for each coordinate, and then, we
could construct $K_{1}=400$ basis functions of two-dimensional observation, $%
\mathbf{X}$, by taking the tensor product. For the VBDM representation, we
first {reduced the data from $MN=8\times$ 640,000 to
$B=B^{1}\times B^{2}=100\times 100$ by the box-averaging method (Appendix~\ref{subsec:data_clus}). Subsequently, we }trained $K_{1}=400$ data-driven basis functions from the $B$ box-averaged data using the VBDM algorithm \cite%
{Berry2016ACHA}.

Figure \ref{Fig2OULike}a displays the analytic conditional density
(\ref{Eqn:station_f2}), and Figures \ref{Fig2OULike}b--d display the pointwise
errors of the conditional densities $%
\widehat{e}\left( \mathbf{X}|\bm{\Sigma }\right) \equiv p\left( \mathbf{X%
}|\bm{\Sigma }\right) -\widehat{p}\left( \mathbf{X}|\bm{\Sigma }%
\right) $ for the training
parameter $(\sigma _{X_{1}}^{2},\sigma _{X_{2}}^{2})=\left( 5,5\right) $.
It can be seen from Figures \ref{Fig2OULike}b--d that all the pointwise errors are small
compared to the analytic $p(\mathbf{X}|\bm{\Sigma })$ in Figure \ref{Fig2OULike}a, so that all
representations of conditional densities $\widehat{p}\left( \mathbf{X}|%
\bm{\Sigma }\right) $ are in excellent
agreement with the analytic $p(\mathbf{X}|\bm{\Sigma })$ (Figure \ref{Fig2OULike}a). This suggests that for the Hermite
representation, the upper bound $D$ (\ref{Eqn:vara_bound}) in Theorem \ref%
{thm:error1} is finite so that the representation is valid in estimating the
conditional density, as can be seen from Figure \ref{Fig2OULike}b. On the other hand, the upper bounds $D$ (\ref{Eqn:vara_bound}) for the cosine and the VBDM representations are always finite, as mentioned in Remark \ref{rmk:finite2} and Proposition \ref{lem:var}, respectively. We should also point out that for this example, the VBDM
representation performed the worst with errors of order $10^{-4}$ compared to the Hermite and cosine representations whose errors were on the order of $10^{-6}$. This larger error in the VBDM representation was because the data-driven basis functions were estimated by discrete eigenvectors $\overline{\bm{\psi }}_{k}\in \mathbb{R}^{B}$, so additional errors \cite{Berry2016ACHA} were introduced through this discrete approximation (especially on the high modes) on the box-averaged data, $\{\overline{\mathbf{y}}_b\}_{b=1,\ldots, B}$, $B=$ 10,000. On the other hand, for Hermite and cosine representations, their analytic basis functions are known, so that the errors could be approximated by (\ref{Eqn:EL2}) in Theorem~\ref{thm:error1}.

\begin{figure*}[ptbh]
\centering \includegraphics[scale=0.8]{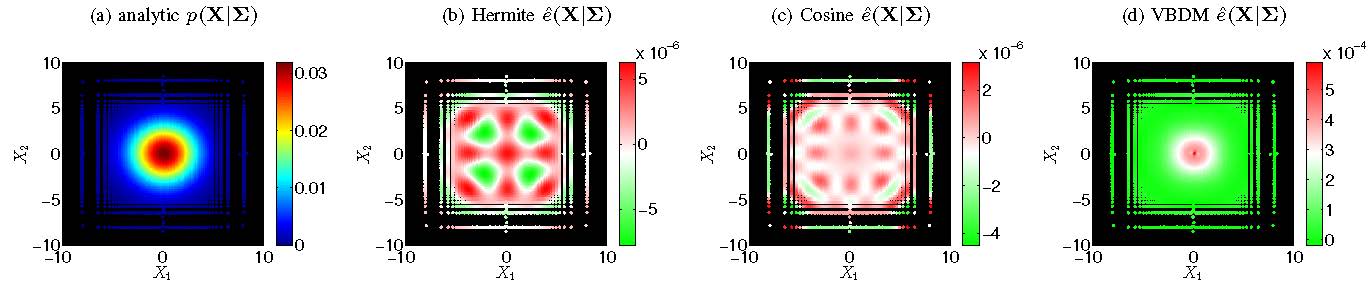}
\caption{(Color online) ({\bf a}) The analytic conditional density $p(\mathbf{X}|\bm{\Sigma})$ (\protect\ref{Eqn:station_f2}).
For comparison, plotted are the pointwise errors of conditional density
functions $\widehat{e}\left( \mathbf{X}|\bm{\Sigma} \right) \equiv
p\left( \mathbf{X}|\bm{\Sigma} \right) -\widehat{p}\left( \mathbf{X}|%
\bm{\Sigma} \right) $ for ({\bf b}) Hermite, ({\bf c}) cosine, and ({\bf d}) VBDM representations.
The density and all the
error functions are plotted on the $B=$ 10,000 box-averaged data
points. The training parameter $\bm{\Sigma} \equiv (\sigma _{X_{1}}^{2},\sigma _{X_{2}}^{2}) = (5,5)$.
}\label{Fig2OULike}
\end{figure*}

We now estimate the posterior density (\ref{Eqn:posterior}) and mean (\ref{Eqn:postmean2}) by using the MCMC method (Procedures
\hyperlink{2A}{(\textit{2-A})}--\hyperlink{2C}{(\textit{2-C})}). We
generated $T=400$ well-sampled normally-distributed data as the
observations from the true values of variance $\bm{\Sigma }^{\dag
}=({(\sigma _{X_{1}}^{2})}^{\dag
},{(\sigma _{X_{2}}^{2})}^{\dag
})=\left( 6.5,6.3\right) $. From the analytical
Formula (\ref{Eqn:postmean2}), we obtained the posterior mean as $(\sigma
_{X_{1}}^{2},\sigma _{X_{2}}^{2})_{PM}=\left( 6.03,5.84\right) $. Here, the
posterior mean deviated greatly from the true value since we only used $T=400$
normally-distributed observation data as new observations. If using much
more new observation data, the analytical posterior mean (\ref{Eqn:postmean2}) will get closer to the true value, $((\sigma
_{X_{1}}^{2})^{\dag
},(\sigma _{X_{2}}^{2})^{\dag
})$. In our simulation, we set the parameter in the prior,
$\alpha _{0}=1$, and the proposal covariance, $\mathbf{C}=0.01\mathbf{I}$.
For each chain, the initial condition $\left( \sigma
_{X_{1},0}^{2},\sigma _{X_{2},0}^{2}\right) $ was drawn randomly from $U[5,12]^{2}$, and
800,000 iterations are generated for the chain.

Figures \ref{Fig3OUMCMC}b,c,d display the densities of the chain
by using Hermite, cosine, and VBDM representation, respectively. The
densities are plotted using the kernel density estimate on the chain ignoring
the first 10,000 iterations. For comparison, Figure \ref{Fig3OUMCMC}a
displays the analytic posterior density (\ref{Eqn:posterior}). It can be
seen from Figure~\ref{Fig3OUMCMC} that the posterior densities by the three
representations were in excellent agreement with each other and with the analytic
posterior density (\ref{Eqn:posterior}). Figure \ref{Fig3OUMCMC} also shows
the comparison between the posterior mean (\ref{Eqn:postmean2}) and the MCMC
mean estimates. From our numerical results, MCMC mean estimates by all
representations and the analytic posterior mean (\ref{Eqn:postmean2}) were
identical within numerical accuracy. Therefore, for this 2D OU-process
example, all representations were valid in estimating the posterior density
and posterior mean of parameter $\bm{\Sigma }$.

\begin{figure}
\centering \includegraphics[scale=1.1]{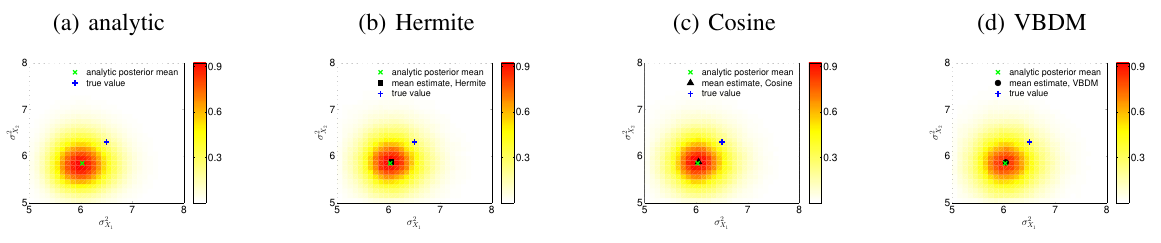}
\caption{(Color online) Comparison of the posterior density functions $p(\bm{\Sigma} |\mathbf{X}^{\dag })$. (a)
Analytical posterior density $p(\bm{\Sigma }|\mathbf{X}^{\dag })$ (\ref{Eqn:posterior}). (b) Hermite
representation. (c) Cosine
representation. (d) VBDM representation. The true value $\bm{\Sigma} ^{\dag } \equiv (\protect\sigma %
_{X_{1}}^{2},\protect\sigma _{X_{2}}^{2})=(6.5,6.3)$ (blue plus). The
analytic posterior mean is $(\protect\sigma %
_{X_{1}}^{2},\protect\sigma _{X_{2}}^{2})_{PM}=(6.03,5.84)$ (green cross).
The MCMC mean estimate using Hermite representation is $(\protect\sigma %
_{X_{1}}^{2},\protect\sigma _{X_{2}}^{2})=(6.05,5.87)$ (black square). The
MCMC mean estimate using Cosine representation is $(\protect\sigma %
_{X_{1}}^{2},\protect\sigma _{X_{2}}^{2})=(6.05,5.87)$ (black triangle). The
MCMC mean estimate using VBDM representation is $(\protect\sigma %
_{X_{1}}^{2},\protect\sigma _{X_{2}}^{2})=(6.04,5.86)$ (black circle).}
\label{Fig3OUMCMC}
\end{figure}

Next, we will investigate a system for
which the intrinsic dimension $d$ of the data manifold where the observations lie is smaller than the dimension of ambient space $n$.

\subsection{\label{sec:torus3d1p} Example II: Three-Dimensional
System of SDE's on a Torus}

\label{Ex2torus} \hypertarget{Ex2}{Consider} a system of SDE's on a torus defined in the
intrinsic coordinates $\left( \theta ,\phi \right) \in \lbrack 0,2\pi )^{2}$:
\begin{equation}
d\left(
\begin{array}{c}
\theta \\
\phi%
\end{array}%
\right) =a\left( \theta ,\phi \right) dt+b\left( \theta ,\phi \right) \left(
\begin{array}{c}
dW_{1} \\
dW_{2}%
\end{array}%
\right) , \label{Eqn:sde_torus}
\end{equation}%
where $W_{1}$ and $W_{2}$\ are two independent Wiener processes,\ and the
drift and diffusion coefficients are:
\begin{eqnarray*}
a\left( \theta ,\phi \right) &=&\left(
\begin{array}{c}
\frac{1}{2}+\frac{1}{8}\cos \left( \theta \right) \cos \left( 2\phi \right) +%
\frac{1}{2}\cos \left( \theta +\pi /2\right) \\
10+\frac{1}{2}\cos \left( \theta +\phi /2\right) +\cos \left( \theta +\pi
/2\right)%
\end{array}%
\right) , \\
b\left( \theta ,\phi \right) &=&\left(
\begin{array}{cc}
D+D\sin \left( \theta \right) & \frac{1}{4}\cos \left( \theta +\phi \right)
\\
\frac{1}{4}\cos \left( \theta +\phi \right) & \frac{1}{40}+\frac{1}{40}\sin
\left( \phi \right) \cos \left( \theta \right)%
\end{array}%
\right) .
\end{eqnarray*}%
The initial condition is $\left( \theta ,\phi \right) =\left( \pi ,\pi
\right) $. Here, $D$ is a parameter to be estimated. This example exhibits
non-gradient drift, anisotropic diffusion, and multiple time scales. {\color{black}Both the observations and the
training dataset were generated by numerically solving the SDE on appropriate parameters $D$} in (\ref%
{Eqn:sde_torus}) with a time step $\Delta t=0.1$ and then mapping this data
into the ambient space, $\mathbb{R}^{3}$, via the standard embedding of the
torus given by:
\begin{equation}
\mathbf{x}\equiv \left( x,y,z\right) =\left( \left( 2+\sin \left( \theta
\right) \right) \cos \left( \phi \right) ,\left( 2+\sin \left( \theta
\right) \right) \sin \left( \phi \right) ,\cos \left( \theta \right) \right)
. \label{Eqn:xyz}
\end{equation}%
Here, $\mathbf{x}\equiv \left( x,y,z\right) $ are observations. This
system on a torus satisfies $d<n$, where $d=2$ is the intrinsic dimension of
$\mathbf{x}$ and $n=3$ is the dimension of ambient space $\mathbb{R}^{n}$.
Our goal is to estimate the posterior density and the posterior mean of
parameter $D$ given discrete-time observations of $\mathbf{x}^{\dag}$, {\color{black}which are the solutions of \eqref{Eqn:sde_torus} for a
specific parameter $D^\dag$.}

For training, we used $M=8$ well-sampled uniformly-distributed training
parameters, $\left\{ D_{j} = j/4\right\} _{j=1}^{8}$.
For each training parameter $D_{j}$, we generated $N=$ 54,000 observations of $%
\mathbf{x}$ {\color{black}by solving the SDE's in \eqref{Eqn:sde_torus} for parameter $D_j$.} For Hermite and cosine representation, we constructed $10$
basis functions for each $x,y,z$ coordinate in Euclidean
space. After taking tensor product of these basis functions, we could obtain $%
K_{1}=1000$ basis functions on the ambient space $\mathbb{R}^{3}$. For VBDM
representation, we first computed $B=B_{1}\times B_{2}\times B_{3}=30^{3}$
box-averaged data points by the data reduction method in Appendix \ref{subsec:data_clus}. However, we found that some
of the $B$ box-averaged data points were far away from the torus. After
ignoring these points, we eventually chose $\widetilde{B}=$ 26,020 of the
box-averaged data points that were close enough to the torus for training.
Then, we trained $K_{1}=1000$ data-driven basis functions on $\mathcal{N}$%
\ from these 26,020\ box-averaged data points using the VBDM algorithm.

Unlike the previous example, the derivation of the analytical expression for the likelihood function $p(\mathbf{x}|D_{j})$ was not trivial. This difficulty is due to the fact that the
diffusion coefficient, $b(\theta,\phi)$, is state dependent. While direct MCMC with an approximate likelihood function constructed using the Bayesian imputation \cite{golightly2010markov} can be done in principle, we neglected this numerical computation since the cost in generating the path $\{\bf{x}_i\}$ for $i=1,\ldots, T$ on each sampling step was too costly in our setup below (where $T=$ 10,000, and we would generate a chain of length 400,000 samples). For diagnostic comparisons, we constructed
another representation $\widehat{p}(\mathbf{x}|D_{j})$, named the intrinsic
Fourier representation, which can be regarded as an accurate approximation
of $p(\mathbf{x}|D_{j})$, as it used the basis functions defined on the intrinsic coordinates $(\theta,\phi)$ instead of $\mathbf{x}\in\mathbb{R}^3$. See Appendix \ref{sec:dsmalN_KDE} for the construction and the convergence of the intrinsic Fourier representation in detail. We should point out that this intrinsic representation is not available in general since one may not know the embedding of the data manifold.

Figure \ref{Fig4TorusLike} displays the comparison of the density estimates.
It can be observed from Figure \ref{Fig4TorusLike} that the VBDM representation
was in good agreement with the intrinsic Fourier representation, whereas Hermite and cosine
representations of $\widehat{p}(\mathbf{x}|D_{j})$\ deviated significantly
from the intrinsic Fourier representation. The reason
in short was that {if the density $p(\theta,\phi|D)$ in $(\theta,\phi)$ coordinate were in $\mathcal{H}([0,2\pi)^2) \subset L^{2}\left( [0,2 \pi)^{2}\right)$, then the corresponding VBDM representation with respect to $dV(\mathbf{x})$
would be in $\mathcal{H}_{\overline{q}^{-1}}(\mathcal{N})$. However, the representation (Hermite and cosine)
with respect to $d\mathbf{x}$, $\mathbf{x} \in \mathbb{R}^{3}$ is not in $\mathcal{H}_{\overline{q}^{-1}}(\mathbb{R}^3)$.
A more detailed explanation of this assertion is presented in Appendix~\ref{sec:dsmalN_KDE}.
\begin{figure}
\centering
\includegraphics[scale=0.87]{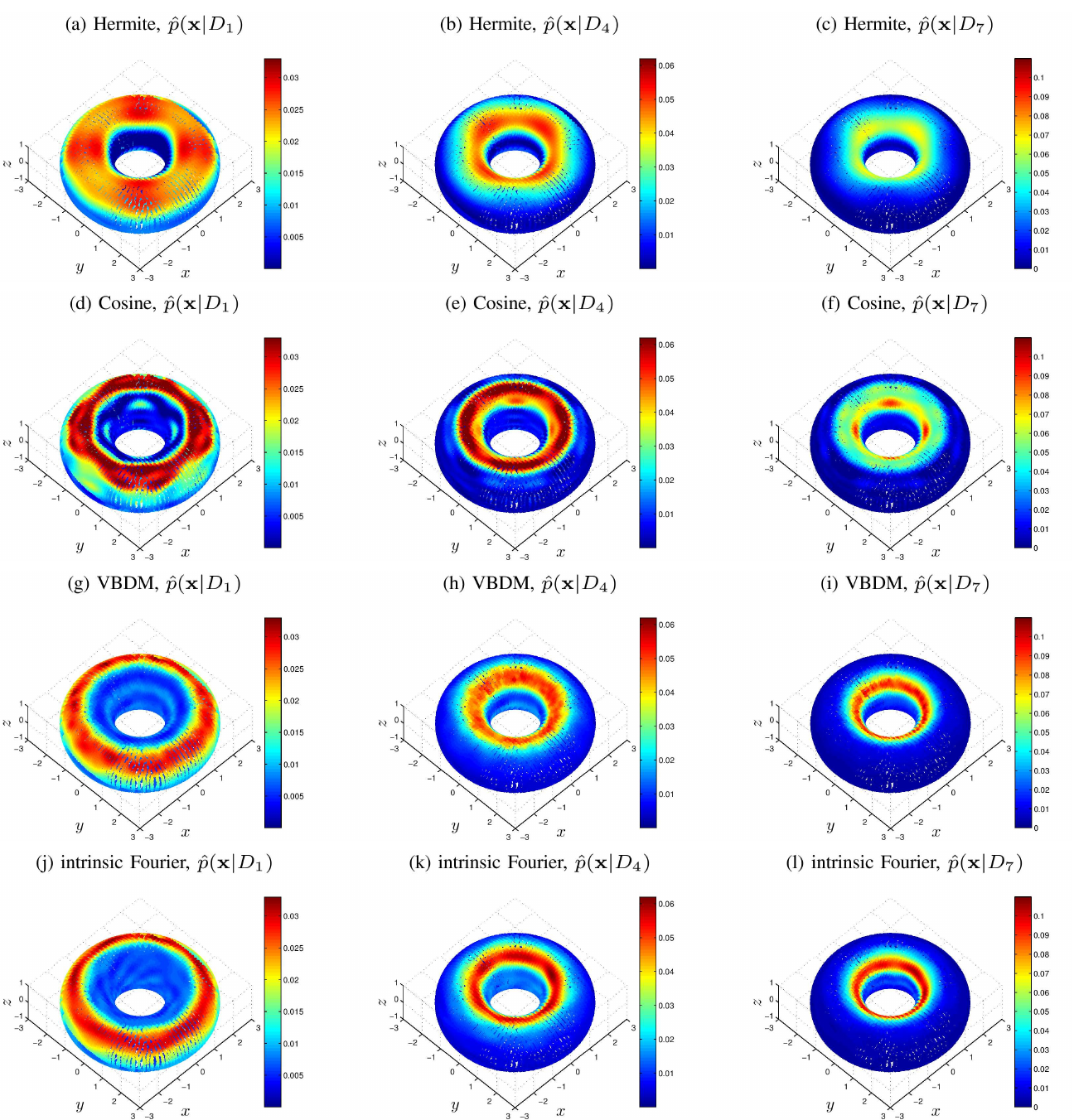}
\caption{(Color online) Comparison of the conditional densities $\hat{p}(%
\mathbf{x}|D_{j})$ estimated by using Hermite representation (first row),
cosine representation (second row), VBDM representation (third row), and
intrinsic Fourier representation (fourth row). The left (a,d,g,j), middle
(b,e,h,k), and right (c,f,i,l) columns correspond to the densities on
the training parameters $D_{1}=0.25$, $D_{4}=1.00$, and $D_{7}=1.75$,
respectively. $K_{1}=1000$ basis functions are used for all representations.
For fair visual comparison,
all conditional densities are plotted on the same box-averaged data points
and normalized to satisfy $\frac{1}{\widetilde{B}}\sum_{b=1}^{%
\widetilde{B}}\widehat{p}(\mathbf{x}_{b}|D_{j})/\overline{q}\left( \mathbf{x}%
_{b}\right) =1$ with $\overline{q}$ being the estimated sampling density of
the box-averaged data $\left\{ \mathbf{x}_{b}\right\} _{b=1}^{\widetilde{B}}$.
}
\label{Fig4TorusLike}
\end{figure}

We now compare the MCMC estimates with the true value, $D^{\dag }=0.9$, from $T=$ 10,000 observations.
For this simulation, we set the prior to be uniformly distributed and empirically chose {${C}=0.01$} for the proposal.
Figure \ref{Fig5TorusMCMC} displays the posterior
densities of the chains for all representations (each plot of the density estimate was constructed using KDE on a chain of length 400,000). Displayed also is the
comparison between the true value $D^{\dag }$\ and the MCMC mean estimates
by all representations. Here, the mean estimate by the intrinsic Fourier
representation nearly overlaps with the true value $D^{\dag }=0.9$, as shown in Figure \ref{Fig5TorusMCMC}. The mean estimate by the VBDM representation is closer to the true value $D^{\dag }$\ compared to the estimates by Hermite and cosine representations.
Moreover, it can be seen from Figure \ref{Fig5TorusMCMC} that the posterior by the
VBDM representation is close to the posterior by intrinsic Fourier
representation, whereas the posterior densities by Hermite and cosine representation are
not. We should point out that this result is encouraging considering that the training parameter domain is rather wide, $D_j\in[1/4,2]$. This result suggests that when the intrinsic dimension is less than the ambient space dimension, $d<n$, the VBDM representation (which does not require the knowledge of the embedding function in \eqref{Eqn:xyz}) with data-driven basis functions in $L^{2}\left( \mathcal{N},\overline{q}\right)$ is superior compared to the representations with analytic basis functions defined on the ambient coordinates $\mathbb{R}^3$.

\begin{figure*}[tbph]
\centering \includegraphics[scale=0.4]{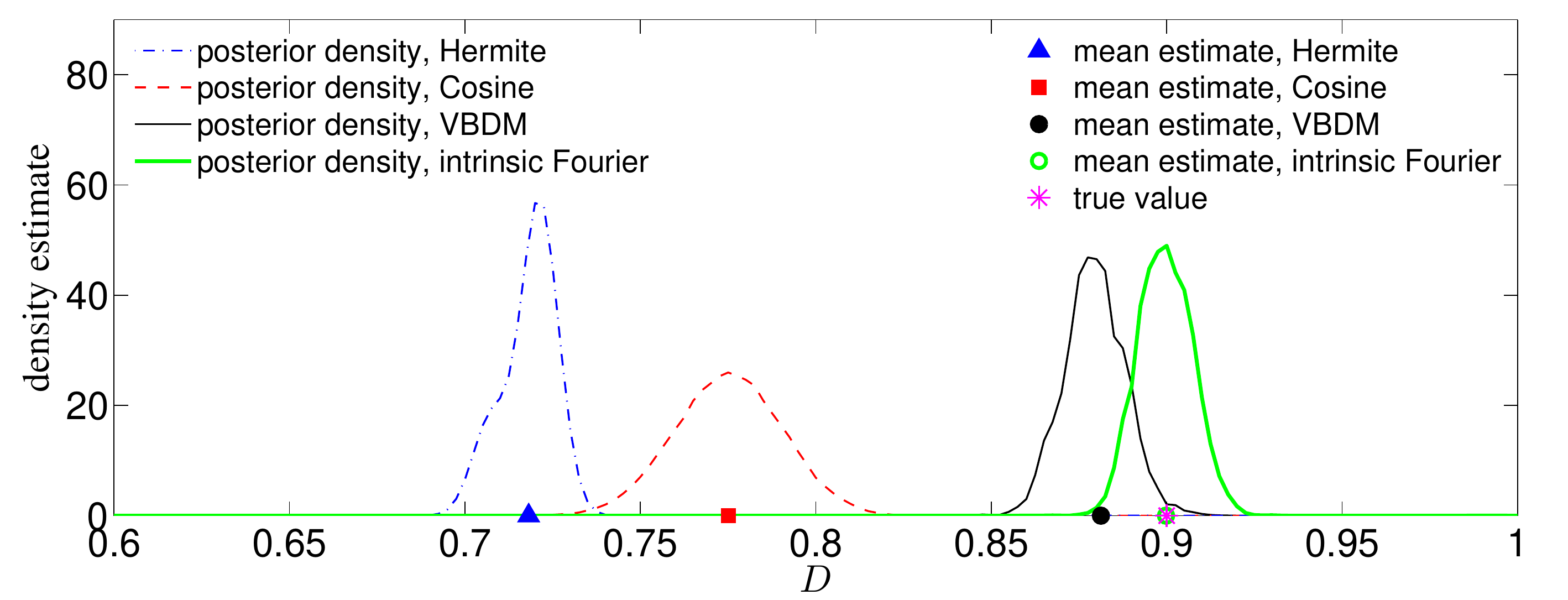}
\caption{(Color online) Comparison of the posterior density functions by all
representations. Plotted also are mean estimates by Hermite representation $%
\hat{D}=0.78$ (blue triangle), cosine representation $\hat{D%
}=0.79$ (red square), VBDM representation $\hat{D}=0.88$
(black circle), the intrinsic Fourier representation $\hat{D}%
=0.90$ (green circle), and the true parameter value $D^{\dag }=0.9 $
(magenta asterisk).}
\end{figure*}

\subsection{\label{sec:l96} Example III: Five-dimensional Lorenz-96
model}

\hypertarget{Ex3}{Consider} the Lorenz-$96$ model
\cite{Lorenz1996}:
\begin{equation}
\frac{dx_{j}}{dt}=x_{j-1}\left( x_{j+1}-x_{j-2}\right) -x_{j}+F,\text{ \ \ }%
j=1,\ldots ,J, \label{Eqn:L96_model2}
\end{equation}%
with periodic boundary, $x_{j+J}=x_{j}$. For the example in this section, we set $J=5$. The initial condition was $x_{j}(0)=\sin (2\pi j/5)$. Our goal here was to estimate the posterior density and posterior
mean of the hidden parameter $F$ given a time series of noisy observations $%
\mathbf{y}^\dagger=\left( y_{1}^\dagger,y_{2}^\dagger,y_{3}^\dagger,y_{4}^\dagger,y_{5}^\dagger\right) $, where:
\begin{equation*}
y_{j}^\dagger\left( t_{m}\right) =x_{j}^\dagger\left( t_{m}\right) +\epsilon _{m,j},\text{ \
\ }\epsilon _{m,j}\sim \mathcal{N}\left( 0,\sigma ^{2}\right) ,\text{ \ }%
m=1,\ldots ,T,
\end{equation*}%
with noise variance $\sigma ^{2}=0.01$. Here, $x_{j}^\dagger\left( t_{m}\right) $ denotes the
approximate solution (with the Runge--Kutta method) with a specific parameter value $F^\dagger$ at discrete times $t_{m}=ms\triangle t$, where $\triangle t=0.05$ is the integration time step and $s$ is the observation interval. Since the embedding function of the observation data is unknown, we do not have a parametric analog to the intrinsic Fourier representation as in the previous example.

In this low-dimensional setting, we can compare the proposed method with basic techniques, including the direct MCMC and the Non-Intrusive Spectral Projection (NISP) method \cite{mx:09}. By direct MCMC, we refer to employing the random walk Metropolis scheme directly on the following likelihood function,
\begin{equation}
p\left( \mathbf{y}^\dagger|F\right) \propto \exp \left( -\frac{\sum_{m=1}^{T}%
\sum_{j=1}^{5}\left( y_{j}^\dagger\left( t_{m}\right) -x_{j}\left( t_{m};F\right) \right) ^{2}}{2\sigma ^{2}}\right) , \label{Eqn:tramcmc_like}
\end{equation}%
where $\sigma ^{2}$ is the noise variance and $x_{j}\left( t_{m};F\right) $ is the solution of the initial value problem in Equation (\ref%
{Eqn:L96_model2}) with the parameter $F$ at time $t_{m}$.
Note that evaluating $x_{j}\left( t_{m};F\right) $ is time consuming if the model time $Ts \Delta t$ is long or the MCMC chain has many iterations. In our implementation, we generated the chain
for 4000 iterations. {This amounts to 4000 sequential evaluations of the likelihood function in \eqref{Eqn:tramcmc_like}, where each evaluation requires integrating the model in \eqref{Eqn:L96_model2} with the proposal parameter value $F^*$ until model unit time $Ts\Delta t$.} We used a uniform prior distribution and {$C=0.1$} for the proposal.

For the NISP method \cite{mx:09}, we used the same Gaussian likelihood function (\ref%
{Eqn:tramcmc_like}) {with approximated $x_j$. In particular, we approximated the solutions $x_j$ with $\tilde{x}_{j}\left( t,F\right) $ for $%
j=1,\ldots ,5$ in the form of:}
\begin{equation}
\tilde{x}_{j}\left( t,F\right) =\sum_{k=1}^{K}\hat{x}_{j,k}\left( t\right)
\varphi _{k}\left( F\right) , \label{Eqn:nisp_expand}
\end{equation}%
where $\varphi _{k}\left( F\right) $ are chosen to be the orthonormal cosine
basis functions, $\hat{x}_{j,k}\left( t\right) $ are the expansion
coefficients, and $K$ is the number of basis functions. Subsequently, we
prescribe a fixed set of nodes $\left\{ F_{j}=7.55+0.1j\right\} _{j=1}^{8}$ {to be used for training $\hat{x}_{j,k}(t)$. Practically, this training procedure only requires
eight model evaluations that can be done in parallel, where each evaluation involves integrating the model with the specified $F_j$ until model unit time $Ts\Delta t$.}
The number of basis functions is $K=8$. After specifying the coefficients $\hat{x}_{j,k}\left( t\right) $ such that $%
\tilde{x}_{j}\left( t,F\right) =x_{j}\left( t;F\right) $, we obtain the
approximation of the solutions $\tilde{x}_{j}\left( t,F\right) $ for all
parameters $F$. {Using these approximated $\tilde{x}_{j}\left( t,F\right) $,
in place of $x_j(t_m,F)$ in \eqref{Eqn:tramcmc_like}, we can generate the Markov chain using the Metropolis scheme.} Again, we used
a uniform prior distribution and {$C=0.1$} for the proposal. In our MCMC implementation, we generated the chain
for 40,000 iterations; {this involved only evaluating \eqref{Eqn:nisp_expand} instead of integrating the true dynamical model in \eqref{Eqn:L96_model2} on the proposal parameter value $F^*$. }

{For RKWHS representations, we also used $M=8$ uniformly-distributed training parameters, $\left\{ F_{j} = 7.55+0.1j \right\} _{j=1}^{8}$. As in the NISP, this training procedure required only eight model integrations with parameter value $F_j$ until the model unit time $Ts\Delta t$, resulting in a total of $MN=8Ts$ training data. In this example, we did not reduce the data using the box-averaging method in Appendix \ref{subsec:data_clus}. In fact, for some cases, such as $s=1$ and $T=50$, the total of training data were only $MN=400$, which was too few for estimation of the eigenfunctions. Of course, one can consider more training parameters to increase this training dataset, but for a fair comparison with NISP,
we chose to just add 10 i.i.d. Gaussian noises to each dataset, resulting in a total of $MN=4000$ for training dataset. This configuration (with a small dataset) is a tough setting for the VBDM since the nonparametric method is advantageous in the limit of a large dataset. When $8Ts$ is sufficiently large, we do not need to increase the dataset by adding multiple i.i.d Gaussian noises. }

For Hermite and Cosine representation, we constructed five Hermite and
cosine basis functions for each coordinate, {which yielding a total of} $%
K_{1}=5^5 = 3125$ basis functions in $\mathbb{R}^{5}$. For the VBDM representation, we
directly applied the VBDM algorithm to train $K_{1}=3125$ data-driven basis
functions on manifold $\mathcal{N}$\ from the $MN=4000$ training dataset. From
the VBDM algorithm, the estimated intrinsic dimension was $d \approx 2$, which
was smaller than the dimension of the ambient space $n=5$.
Then, we applied a uniform prior distribution and {$C=0.01$} for the proposal. {As in NISP, we generated the chain for 40,000 iterations, which amounted to evaluating \eqref{Eqn:new_like} instead of integrating the true dynamical model in \eqref{Eqn:L96_model2} on each iteration.}

We now compare the posterior densities and mean estimates for
the case of $s=1$ and $T=50$ noisy observations $\mathbf{y}^\dagger\left( t_{m}\right) $
corresponding to the true parameter value $F^{\dag }=8$.
Figure~\ref{Fig7_L96_Like} displays the posterior densities of the chains and mean estimates
for the direct MCMC method,
NISP method, and all
representations.
It can be
seen from Figure \ref{Fig7_L96_Like} that the mean estimate by VBDM
representation was in good agreement with the true value $F^{\dag }$.
In contrast, the mean estimates
by Hermite and cosine representations deviated substantially from the true value.
Based on this numerical result,
where the estimated intrinsic dimension $d\approx 2$ of the observations was lower than the ambient space dimension $n=5$,
the data-driven VBDM representation was superior compared to the Hermite and cosine representations.
It can be further observed that direct MCMC, NISP, and VBDM
representation can provide good mean estimates to the true value.
However, notice that we only ran the model $M=8$ times for the NISP method and VBDM representation, whereas
we ran the model 4000 times for the direct MCMC method.

\begin{figure*}[tbph]
\centering  \includegraphics[scale=0.4]{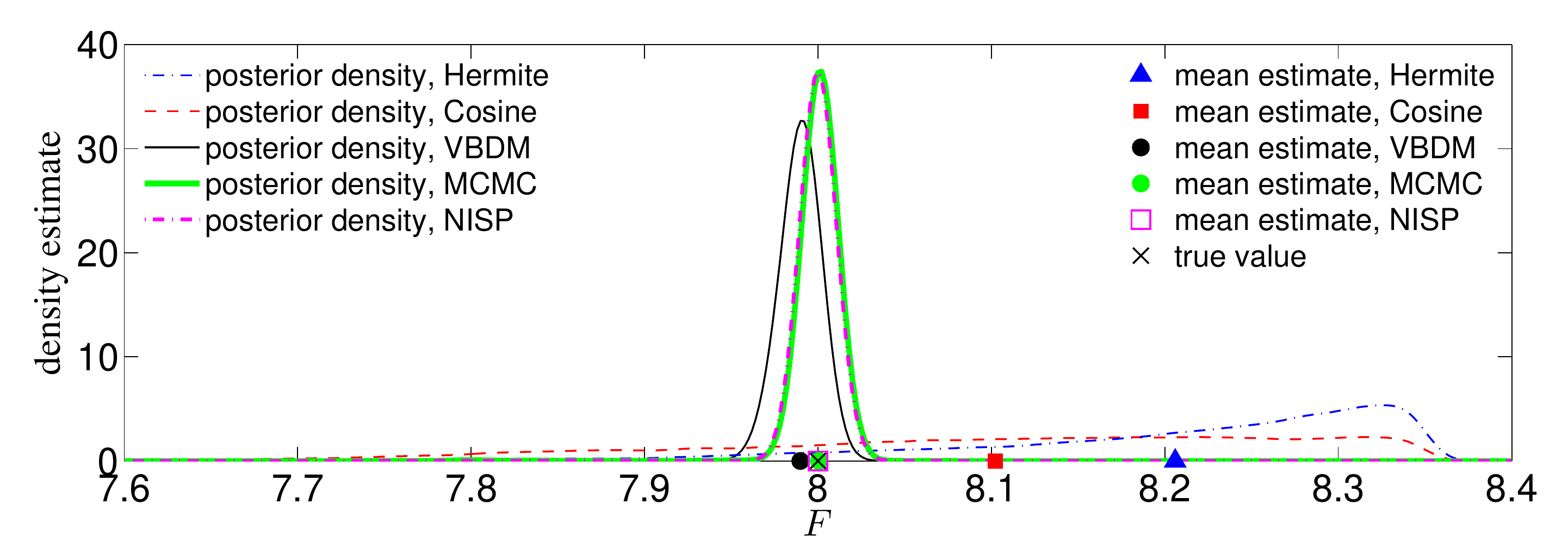}
\caption{(Color online) Comparison of the posterior density functions among the direct MCMC method,
NISP method, and all
RKWHS representations. Plotted also are the true parameter value $F^{\dag }=8 $ (black
cross), mean the estimate by the direct MCMC method $\hat{F}=8.00$ (green circle), the mean estimate by the
NISP method $\hat{F}=8.00$ (magenta square), and the mean estimates by Hermite representation $%
\hat{F}=8.21$ (blue triangle), cosine representation $\hat{F%
}=8.10$ (red square), and VBDM representation $\hat{F}=7.99$
(black circle). The noisy observations are $y_{j^\dagger}(t_{m})$ for $s=1$, $T=50$.
}\label{Fig7_L96_Like}
\end{figure*}

{ In real applications where the observations are not simulated by the model, we expect the observation configuration to be pre-determined. Therefore it is important to have an algorithm that is robust under various observation configurations. In our next numerical experiment, we checked such robustness by comparing the direct MCMC method,
NISP method, and VBDM representation for different cases of $s$ and $T$ (Figure \ref{Fig8L96MCMC}a). }
It can be observed from Figure \ref{Fig8L96MCMC}a that both the direct MCMC method and VBDM representation
can provide reasonably accurate mean estimates for all cases of $s$ and $T$.
However, again notice that we need to run the model much more times for the direct MCMC method
than for VBDM representation.
It can be further observed that the NISP method can only provide a good mean estimate for observation time up to
$Ts\Delta t = 200\Delta t$ when eight uniform nodes $\left\{ F_{j} = 7.55+0.1j \right\} _{j=1}^{8}$ are used. The reason was that
the approximated solution by NISP method was only accurate for observation time up to $200\Delta t$ (see the green and red curves in Figure \ref{Fig8L96MCMC}b). {This result suggests that our surrogate modeling approach using the VBDM representation can provide accurate and robust mean estimates under various observation configurations.}

\begin{figure*}[tbph]
\centering \includegraphics[scale=1.2]{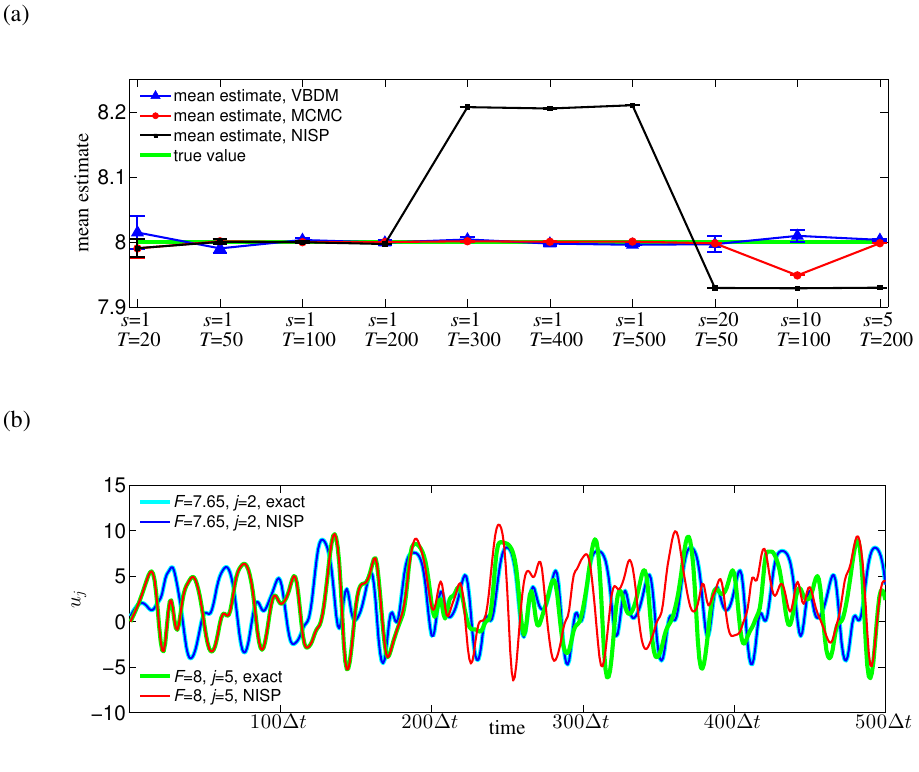}
\caption{(Color online) ({\bf a}) Comparison of the mean estimates among the direct MCMC method,
NISP method, and VBDM
representation for different cases of $s$ and $T$. Plotted also is the true parameter value $F^{\dag }=8 $ (green
curve). ({\bf b}) Comparison of the exact solution by numerical integration and the approximated solution by the NISP method
at the training parameter $F=7.65$ and at the parameter value $F=8$, which is not in the training parameter.}
\label{Fig8L96MCMC}
\end{figure*}

\subsection{Example IV: The 40-dimensional Lorenz-96 model}

In this section, we consider estimating the parameter $F$ in the Lorenz-96 model in \eqref{Eqn:L96_model2}, but of a $J=40$ dimensional system. We now consider observing the autocorrelation function of several energetic Fourier modes
of the system phase-space variables. In particular, let $\{\hat{x}_k(t_m;F)\}_{k=-J/2+1,\ldots,J/2}$ be the $k^{\text{th}}$ discrete Fourier mode of $\{x_j(t_m;F)\}_{k=1\ldots,J}$, where $t_m=m\Delta t$ with $\Delta t=0.05$. Let the observation function be defined as in \eqref{generalobsmodel} with four-dimensional $\{\mathbf{y}_{m}(F)\}_{m=0,\ldots,T}$, whose components are the autocorrelation function of Fourier mode $k_j$,
\BEA
y_{m,j}(F) = \mathbb{E}[\hat{x}_{k_j}(t_m;F)\hat{x}_{k_j}(t_0;F)], \quad m= 0,\ldots, T, \quad j=1,\ldots,4, \nonumber
\EEA
of the energetic Fourier modes, $k_j \in \{7,8,9,14\}$. See \cite{mag:05} for the detailed discussion of the statistical equilibrium behavior of this model for various values of $F$. Such observations arise naturally since some of the model parameters can be identified from non-equilibrium statistical information via the linear response statistics \cite{hlz:17,zlh:19}. In our numerics, we will approximate the correlation function by averaging over a long trajectory,
\BEA
 \mathbb{E}[\hat{x}_{k_j}(t_m;F)\hat{x}_{k_j}(t_0;F)] \approx \frac{1}{L} \sum_{\ell=1}^L \hat{x}_{k_j}(t_{m+\ell};F)\hat{x}_{k_j}(t_\ell;F),\label{autocor}
\EEA
with $L=10^6$. Here, each of these Fourier modes is assumed to have zero empirical mean. We will consider observing the autocorrelation function up to time $T=50$ (corresponding to 2.5 unit time).

With this setup, the corresponding likelihood function for $p(\mathbf{y}_m|F)$ is not easily approximated (since it is not in the form of \eqref{Eqn:trad_like}), and it is computationally demanding to generate $y_{m,j}(F)$ since each evaluation requires integration of the 40-dimensional Lorenz-96 model up to time index $L=10^6$. This expensive computational cost makes either the direct MCMC or approximate Bayesian computation infeasible. We should also point out the fact that a long trajectory is needed in the evaluation of \eqref{autocor}, making this problem intractable with NISP even if a parametric likelihood function becomes available. This issue is because the approximated trajectory by polynomial chaos expansion in NISP is only accurate for short times, as shown in the previous example. We will consider constructing the likelihood function from a wide range of training parameter values, $F_i = 6+0.1(i-1), i=1,\dots, M=31$. This parameter domain is rather wide and includes the weakly chaotic regime ($F=6$) and strongly chaotic regime ($F=8$). See \cite{am:07} for a complete list of chaotic measures in these regimes including the largest Lyapunov exponent and the Kolmogorov--Sinai entropy.

In this setup, we had a total of $MN=M(T+1) = 31\times 51 = 1581$ of $\mathbf{y}_m(F_i)\in\mathbb{R}^4$ for training. We will consider an RKHS representation with $K_1=500$ basis functions. We will demonstrate the performance on 30 sets of observations $\mathbf{y}_m(F^\dagger_s)$, where in each case, $F^\dagger_s$ does not belong to the training parameter set, namely $F^\dagger_s = 6.05+0.1(s-1), s=1,\ldots,30$. In each simulation, the MCMC initial chain will be set to be random, $F \sim \mathcal{U}(6.5,8.5)$; the prior is uniform; and $C=0.01$ for the proposal. In Figure~\ref{Fig9L96MCMC}, we show the mean estimates and an error bar (based on one standard deviation) computed from averaging the MCMC chain of length 40,000 in each case. Notice the robustness of these estimates on a wide range of true parameter values $F^\dagger$ using a likelihood function constructed using a single set of training parameter values on $[6,9]$.

\begin{figure*}[tbph]
\centering \includegraphics[width=0.7\textwidth]{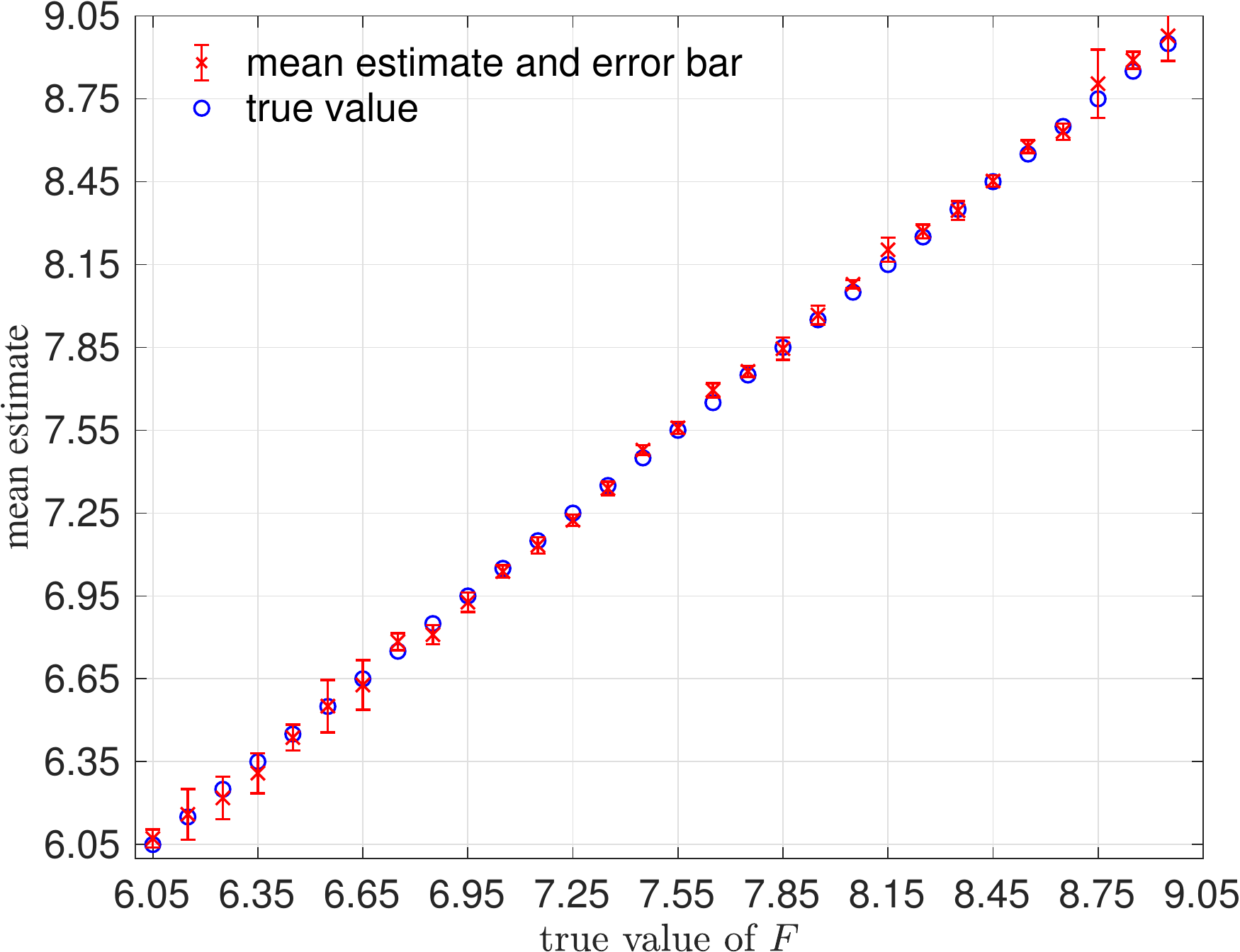}
\caption{(Color online) Mean error estimates and error bars for various true values of $F$ that are not in the training parameters.}
\label{Fig9L96MCMC}
\end{figure*}

\section{\label{sec:conclus}Conclusion}

We have developed a framework of a parameter estimation approach where MCMC was
employed with a nonparametric likelihood function. Our approach approximated the
likelihood function using the kernel embedding of conditional distribution
formulation based on RKWHS. By analyzing the error estimation in Theorem \ref%
{thm:error1}, we have verified the validity of our RKWHS representation of
the conditional density as long as $p(\mathbf{y}|\bm{\theta
}_{j})\in \mathcal{H}_{q^{-1}} \left( \mathcal{N}\right) $ induced by the basis in ${L}^2(\mathcal{N},q)$
and Var$_{\mathbf{Y}|%
\bm{\theta }_{j}}\left[ \psi _{k}(\mathbf{Y})\right] $ is finite. Furthermore,
the analysis suggests that if the weight $q$ is chosen to be the sampling density of
the data, the Var$_{\mathbf{Y}|%
\bm{\theta }_{j}}\left[ \psi _{k}(\mathbf{Y})\right] $ is always finite. This justifies the
use of Variable Bandwidth Diffusion Maps (VBDM) for estimating the data-driven basis functions of
the Hilbert space weighted by the sampling density on the data manifold.

We have demonstrated the proposed approach with four numerical examples. In the first example, where the dimension of the data manifold was exactly the dimension of the ambient space, $d=n$, the RKHS representation with VBDM basis yielded a parameter estimate as accurate as using other analytic basis representation. However, in the examples where the dimension of the data manifold was strictly less than the dimension of the ambient space, $d<n$, only VBDM representation could provide more accurate estimation of the true parameter value. We also found that VBDM representation produced mean estimates that were robustly accurate (with accuracies that were comparable to the direct MCMC) on various observation configurations where the NISP was not accurate. This numerical comparison was based on using only eight model evaluations, which can be done in parallel for both VBDM and NISP, whereas the direct MCMC involved 4000 sequential model evaluations. Finally, we demonstrated robust accurate parameter estimation on an example where the analytic likelihood function was intractable and computationally demanding, even if it became available. Most importantly, this result was based on training on a wide parameter domain that included different chaotic dynamical behaviors.

From our numerical experiments, we conclude that the proposed nonparametric representation was advantageous in any of these configurations: (1) when the parametric likelihood function was not known, such as in Example IV; (2) when the observation time stamp was long (such as in Example II or for large $sT$ in Example III and Example IV). Ultimately, the only real advantage of this method (as a surrogate model) was when the direct MCMC or ABC, which require sequential model evaluation, was computationally not feasible.

While the theoretical and numerical results were encouraging as a proof the concept for using the VBDM representation in many other parameter estimation applications, there were still practical limitations that need to be overcome. As in the other surrogate modeling approaches, one needs to have knowledge of the feasible domain for the parameters. Even when the parameter domain is given and wide, it is practically not feasible to generate training dataset by evaluating the model on the specified training grid points on this domain when the dimension of the parameter space is large (e.g., order 10), even if the Smolyak sparse grid is used. One possible way to simultaneously overcome these two issues is to use ``crude'' methods, such as ensemble Kalman filtering or smoothing, to obtain the training parameters. We refer to such a method as ``crude'' since the parameter estimation with ensemble Kalman filtering is sensitive to the initial conditions, especially when the persistent model is used as the dynamical model for the parameters \cite{harlim2018}. However, with such crude methods, we can at least obtain a set of parameters that reflect the observational data, instead of specifying training parameters uniformly or in a random fashion, which can lead to unphysical training parameters. Another issue that arises in the VBDM representation is the expensive computational cost when the amount of data $MN$ is large. When the dimension of the {observations} is low (as in the examples in this paper), the data reduction technique described in Appendix \ref{subsec:data_clus} is sufficient. For larger dimensional problems, a more sophisticated data reduction is needed.
Alternatively, one can explore representations using other orthonormal data-driven basis, such as the QR factorized basis functions as a less expensive alternative to the eigenbasis \cite{Harlim2017diffusion}.

\section*{Acknowledgments}
This research of is partially supported by the ONR Grant N00014-16-1-2888. J.H would also acknowledge supports from the NSF Grant DMS-1619661. 

\appendix

\section{\label{subsec:data_clus}Data {reduction}}

When $MN$ is very large, the VBDM algorithm
becomes numerically expensive since it involves solving an eigenvalue problem of matrix size $MN\times MN$.
Notice that the number of training parameters $M$ grows {exponentially as a function of}
 the dimension of parameter, $m$, if well-sampled uniformly distributed training parameters are used. To overcome this
large training data problem, we {employ an empirical} data reduction method
to reduce the original $MN$ training data points $\left\{ \mathbf{y}%
_{i,j}\right\} _{j=1,\ldots ,M}^{i=1,\ldots ,N}$\ to a small $B\left( \ll
MN\right) $ number of training data points {yet preserving}
the sampling density $\overline{q}\left( \mathbf{y}\right) $ in (\ref%
{Eqn:q_bar_y}). Subsequently, we apply the VBDM algorithm on these reduced
training data points. {It is worthwhile to mention that this data reduction method
is numerically applicable for low-dimensional dataset although in the following we will introduce this reduction method
for any $n$-dimensional dataset.}

The basic idea of our method is that we first cluster the training dataset $%
\left\{ \mathbf{y}_{i,j}\right\} _{j=1,\ldots ,M}^{i=1,\ldots ,N}$\ into $B$
number of boxes and then take the average of data points in each box as a
reduced training data point. First, we cluster the training data $\left\{
\mathbf{y}_{i,j}\right\} $,\ based on the ascending order of the $1^{\text{st%
}}$ coordinate of the training data, $\left\{ y_{i,j}^{1}\right\} $, into $%
B^{1}$ number of groups such that each group has the same number $\left(
=MN/B^{1}\right) $ of data points. After the first clustering, we obtain $%
B^{1}$ groups with each group denoted by $G_{k_{1}}^{1}$ for $k_{1}=1,\ldots
,B^{1}$. Here, the super-index $1$ denotes the first clustering and the
sub-index $k_{1}$ denotes the $k_{1}^{\text{th}}$ group. Second, for each group $%
G_{k_{1}}^{1}$, we cluster the training data $\left\{ \mathbf{y}%
_{i,j}\right\} $ inside $G_{k_{1}}^{1}$,\ based on the ascending order of
the second coordinate of the training data, $\left\{
y_{i,j}^{2}\right\} $, into $B^{2}$ number of groups such that each group
has the same number $\left( =MN/B^{1}B^{2}\right) $\ of data points. After
the second clustering, we obtain totally $B^{1}B^{2}$ groups with each group
denoted by $G_{k_{1}k_{2}}^{2}$ for $k_{1}=1,\ldots ,B^{1},$ and $%
k_{2}=1,\ldots ,B^{2}$. We can operate such clustering $n$ times, where $n$
is the dimension of the observation $\mathbf{y}$ ambient space. After $n$ times
clustering, we obtain $B\equiv \prod_{s=1}^{n}B^{s}$ groups with each group
denoted by $G_{k_{1}k_{2}\ldots k_{n}}^{n}$ with $k_{s}=1,\ldots ,B^{s},$
for all $s=1,\ldots ,n$. Each group is a box [see Figure \ref{Fig1datasort}
for example]. After taking the average of the data points in each box $%
G_{k_{1}k_{2}\ldots k_{n}}^{n}$, we obtain $B$ number of reduced
training data points. In the remainder of this paper, we denote these $B$
number of reduced training data points by $\left\{ \overline{\mathbf{y}}%
_{b}\right\} _{b=1,\ldots ,B}$ and refer to them as the box-averaged data
points. Intuitively, this algorithm partitions the domain into hyperrectangle such that $P_{r}(\overline{\mathbf{y}}\in
G_{k_{1}\ldots k_{n}}^{n})\approx 1/B$. {Note that the idea of our data reduction method is analogous to that of multivariate $k$-nearest neighbor density estimates \cite{loftsgaarden1965nonparametric,mack1979multivariate}.} The error estimation can be found in the Refs. \cite{loftsgaarden1965nonparametric,mack1979multivariate}.

To examine whether the distribution of these box-averaged data points is
close to the sampling density $\overline{q}\left( \mathbf{y}\right) $ of the
original dataset, we apply our reduction method to several numerical examples in the
following. Figure \ref{Fig1datasort}a shows the reduction
result for few number ($=64$) of uniform data in $[0,1]\times \lbrack 0,1]$.
Here, $B^{1}=4$ and $B^{2}=4$\ so that there are total $B=16$ boxes and
inside each box there are $4$ uniform data points (blue circles).
It can be
seen that the box-averaged data points (red circles) are far away from the
well-sampled uniform data points (cyan crosses). However, when uniform data
points (blue circles) increase to a large number ($=6400$), box-average data
points (red circles) are very close to well-sampled uniform data points
(cyan crosses) as shown in Figure \ref{Fig1datasort}b. This suggests that
these box-averaged data points nearly admit the uniform distribution when there are a large number of original
uniform data points. Figures \ref{Fig1datasort}c and d show the comparison
of the kernel density estimates applied on the box-averaged data for different $B$ for
standard normal distribution and the distribution proportional to $\exp %
\left[ -\left( -X_{1}^{2}+X_{1}^{3}+X_{1}^{4}\right) \right] $,
respectively. It can be seen that the reduced box-averaged data points
nearly preserve the distribution of the original large dataset, $N=640,000$.

\begin{figure*}[ptbh]
\centering \includegraphics[scale=1.3]{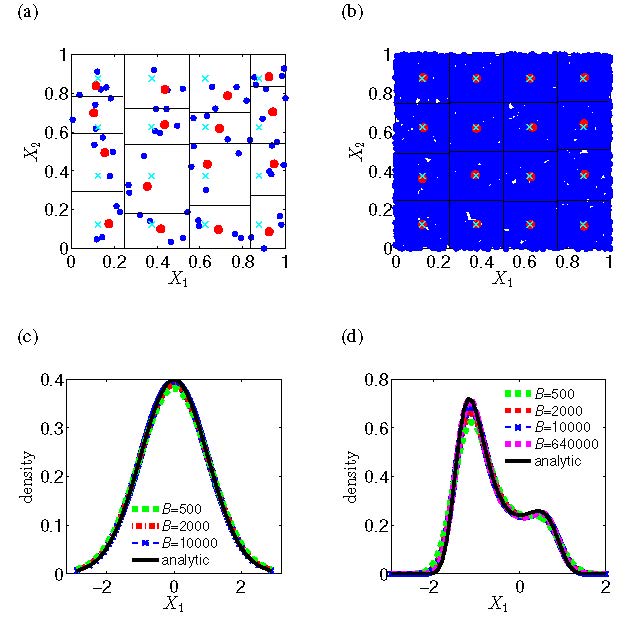} %
\caption{(Color online) Data reduction for ({\bf a}) few number ($=64$) of
uniformly distributed data, and ({\bf b}) many number ($=6400$) of uniformly
distributed data. The $64$ blue circles correspond to uniformly distributed data, $16$
cyan crosses correspond to well-sampled uniformly distributed data, and $16$
red circles correspond to box-averaged data. Boxes are partitioned by
horizontal and vertical black lines. The vertical black lines correspond to the first
clustering and the horizontal lines correspond to the second clustering.
{Panels ({\bf c}) and ({\bf d}) display} the comparison
of kernel density estimates on the box-averaged data for different number $B$
for (c) standard normal distribution, and (d) the distribution proportional
to $\exp[-(-X_{1}^{2}+X_{1}^{3}+X_{1}^{4})]$, respectively. For comparison, also plotted
is the analytic probability density of the distribution. The total number of the points is $640,000$. It can be seen that
the reduced box-averaged data points nearly preserve the distribution of the
original dataset. }
\label{Fig1datasort}
\end{figure*}

When $MN$ is very large, the VBDM algorithm for the construction of
data-driven basis functions [Procedure \hyperlink{1B}{(\textit{1-B})}] can
be outlined as follows. We first use our data reduction method to obtain $%
B\left( \ll MN\right) $ number of box-averaged data $\left\{ \overline{%
\mathbf{y}}_{b}\right\} _{b=1,\ldots ,B}\subseteq \mathcal{N}\subseteq
\mathbb{R}^{n}$ with sampling density $\overline{q}\left( \mathbf{y}\right)
\approx \frac{1}{M}\sum_{j=1}^{M}p(\mathbf{y}|\bm{\theta }_{j})$ in (\ref%
{Eqn:q_bar_y}). The sampling density $\overline{q}$ is estimated at the
box-averaged data $\overline{\mathbf{y}}_{b}$ using all the box-averaged
data points $\left\{ \overline{\mathbf{y}}_{b}\right\} _{b=1,\ldots ,B}$\ by
a kernel density estimation method. Implementing the VBDM algorithm, we can
obtain orthonormal eigenvectors $\overline{\bm{\psi }}_{k}\in \mathbb{R}^{B}$%
, which are discrete estimates of the eigenfunctions $\overline{\psi }%
_{k}\left( \mathbf{y}\right) \in L^{2}\left( \mathcal{N},\overline{q}\right)
$. The $b^{\text{th}}$ component of the eigenvector $\overline{\bm{\psi }}_{k}$ is a
discrete estimate of the eigenfunction $\overline{\psi }_{k}\left( \overline{%
\mathbf{y}}_{b}\right) $, evaluated at the box-averaged data point $%
\overline{\mathbf{y}}_{b}$. Due to the dramatic reduction of the training
data, the computation of these eigenvectors $\overline{\bm{\psi }}_{k}\in
\mathbb{R}^{B}$ becomes much cheaper than the computation of the
eigenvectors $\overline{\bm{\psi }}_{k}\in \mathbb{R}^{MN}$ using the
original training dataset $\left\{ \mathbf{y}_{i,j}\right\} _{j=1,\ldots
,M}^{i=1,\ldots ,N}$. Then we can obtain a discrete representation (\ref%
{Eqn:pyb_repre}) of the conditional density at the box-averaged data points $%
\overline{\mathbf{y}}_{b}$, $\widehat{p}\left( \overline{\mathbf{y}}_{b}|%
\bm{\theta }\right) $.

\section{\label{sec:dsmalN_KDE}Additional Results on Example~{\hyperref[Ex2torus]{II} }}

{In this section, we discuss the intrinsic Fourier representation constructed for numerical comparisons in Example~{\hyperref[Ex2torus]{II} } and provide more {detailed} discussion on the numerical results.}

We first discuss the construction for the intrinsic Fourier representation
of the {true conditional density, $p\left( \mathbf{x}|D\right) $,} defined with respect to the volume form inherited by $\mathcal{N}$ from the ambient space $\mathbb{R}^n$ for
the system on the torus (\ref{Eqn:sde_torus}) in Example \hyperlink{Ex2}{ II}. By noticing the embedding (\ref{Eqn:xyz}) of $\left( \theta ,\phi \right) $
in $\mathbf{x}\equiv \left( x,y,z\right) $, we can obtain the following
equality,

\begin{eqnarray}
1 &=&\int_{\mathcal{N}}p\left( \mathbf{x}|D\right) dV\left( \mathbf{x}\right)
=\int_{[0,2\pi )^{2}}p\left( \mathbf{x}\left( \theta ,\phi \right)
|D\right) \left\vert \mathbf{x}_{\theta }\times \mathbf{x}_{\phi
}\right\vert d\theta d\phi \notag \\
&\equiv& \int_{[0,2\pi )^{2}}p_{\text{IC}}\left( \theta ,\phi
|D\right) d\theta d\phi , \label{Eqn:pic_trans}
\end{eqnarray}%
where $dV\left( \mathbf{x}\right) =\left\vert \mathbf{x}_{\theta }\times
\mathbf{x}_{\phi }\right\vert d\theta d\phi $ is the volume form, $p_{\text{IC}}$
denotes the {true conditional density as a function of the intrinsic coordinates, $\left( \theta ,\phi \right) $. Assuming that $p_{\text{IC}}\left( \theta ,\phi
|D\right)\in \mathcal{H}([0,2\pi)^2) \subset L^2([0,2\pi)^2)$,
and the relation in (\ref{Eqn:pic_trans}), we can construct the
intrinsic Fourier representation as follows,}
\begin{equation}
\widehat{p}\left( \mathbf{x}|D\right) =\frac{\widehat{p}_{\text{IC}}\left(
\theta ,\phi |D\right) }{\left\vert \mathbf{x}_{\theta }\times
\mathbf{x}_{\phi }\right\vert }, \label{Eqn:rel_pic_px}
\end{equation}%
where $\widehat{p}_{\text{IC}}\left( \theta ,\phi |D\right) $ is a
RKWHS representation (\ref{Eqn:pyb_repre}) of the conditional density $%
p_{\text{IC}}\left(\theta ,\phi |D\right) $ with a set of
orthonormal Fourier basis functions $\psi _{k}\left(\theta ,\phi
\right) \in L^{2}\left( [0,2\pi )^{2}\right) $. Here, $\psi
_{k}\left(\theta ,\phi \right) $ are formed by the tensor
product of two sets of orthonormal Fourier basis functions $\Big\{1$, $%
\left\{ \sqrt{2}\cos \left( m\theta \right) \right\}$
, $\left\{ \sqrt{2}\sin \left( m\theta \right) \right\} \Big\}$ and $\Big\{1$, $\left\{ \sqrt{2}\cos \left( m\phi \right)
\right\}$, $\left\{ \sqrt{2}\sin \left( m\phi \right)
\right\}\Big\}$ for $m\in \mathbb{N}^{+}$. Note that for intrinsic Fourier
representation, we need to know the embedding (\ref{Eqn:xyz}) and know the
data of $\left( \theta ,\phi \right) $ in intrinsic coordinates for
training, {which is available for this example.} Nevertheless, for Hermite, Cosine, and VBDM representations, we
only need to know the observation data $\mathbf{x}$ for training.

The convergence of $\widehat{p}\left( \mathbf{x}|D\right) $ to the true density can be explained as follows.
For the system (\ref{Eqn:sde_torus}) in the intrinsic coordinates $\left(
\theta ,\phi \right) $, {where} $p_{\text{IC}}\left( \theta ,\phi
 |D\right) \in \mathcal{H}\left( [0,2\pi )^{2}\right) $ for all parameter $%
D $, the statistics Var$_{\left( \theta ,\phi \right) |D}\left[ \psi _{k}\left( \theta
,\phi \right) \right] $ are bounded for all $D$
and all $k\in \mathbb{N}^{+}$, {by the compactness of $[0,2\pi )^{2}$ and the uniform boundedness of $\psi _{k}\left( \theta
,\phi \right)$ for all $k$}. According to Theorem \ref{thm:error1}, we can obtain the convergence of
the representation $\widehat{p}_{\text{IC}}\left( \theta ,\phi
|D\right) $. Then by noticing the smoothness of $\left\vert \mathbf{x}%
_{\theta }\times \mathbf{x}_{\phi }\right\vert $ on the torus, we can obtain
the convergence of the intrinsic Fourier representation $\widehat{p}\left(
\mathbf{x}|D\right) $\ in (\ref{Eqn:rel_pic_px}).

Next, we give an intuitive explanation for the reason why in the regime $d<n$,
VBDM representation can provide a good approximation whereas Hermite and
Cosine representations cannot. Essentially, the VBDM representation
uses basis functions of the weighted Hilbert space of functions defined with respect to a volume form $\widetilde{V}$ that is conformally equivalent to the volume form $V$ that is inherited by the data manifold $\mathcal{N}$ from the ambient space, $\mathbb{R}^n$. That is, the weighted Hilbert space,
 $L^2(\mathcal{N},\overline{q}^{-1})$, means,
\begin{equation}
L^2(\mathcal{N},\overline{q}^{-1}) = \Big\{ f : \int_{\mathcal{N}} | f(\mathbf{x})|^2 d\widetilde V(\mathbf{x})<\infty\Big\},\nonumber
\end{equation}
where $d\widetilde{V}(\mathbf{x}) = \overline{q}(\mathbf{x})^{-1} dV(\mathbf{x})$ denotes the volume form that is conformally changed by the sampling density $\overline{q}$. We should point out that the key point of the diffusion maps algorithm \cite{Coifman2006ACHA} is to introduce an appropriate normalization to avoid biased in the geometry induced by the sampling density $\overline{q}$ when the data are not sampled according to the Riemannian metric inherited by $\mathcal{N}$ from the ambient space $\mathbb{R}^n$. Furthermore, the orthonormal basis functions of the Hilbert space $L^2(\mathcal{N},\overline{q}^{-1})$ are the eigenfunctions of the adjoint (with respect to $L^2(\mathcal{N})$) of the operator, $\mathcal{L}=\nabla\log(\overline{q})\cdot\nabla + \Delta$, that is constructed by the VBDM algorithm. Incidentally, the adjoint operator $\mathcal{L}^*$ is the Fokker-Planck operator of a gradient system forced by stochastic noises. The point is that this adjoint operator takes density functions of the weighted Hilbert space $L^2(\mathcal{N},\overline{q}^{-1})$. Since the Hilbert space $L^{2}\left( \mathcal{N},\overline{q}^{-1}\right)$ is a function space of some Fokker-Planck operator that acts on densities defined with respect to the geometry of data, then representing the conditional density with basis functions of the weighted Hilbert space $L^{2}\left( \mathcal{N},\overline{q}^{-1}\right)$ is a natural choice. Thus, the error estimation in Theorem \ref{thm:var2}\ is valid in controlling the error of the estimate.

Next, we will show that the representation of the true density {$p_{\text{EX}}$} in the ambient space $\mathbb{R}^3$ is not a function of
$\mathcal{H}_{q^{-1}}(\mathcal{R})\subset L^{2}\left( \mathcal{R},q^{-1}\right)$, where for Hermite representation $\mathcal{R}$ is $\mathbb{R}^{3}$ and $q$ is a normal distribution, and for Cosine representation $\mathcal{R}$ is a hyperrectangle containing the torus and $q$ is a uniform distribution. Recall that the torus is parametrized by:
\begin{equation}
\mathbf{x}\equiv \left( x,y,z\right) =\left( \left( 2+r\sin \left( \theta
\right) \right) \cos \left( \phi \right) ,\left( 2+r\sin \left( \theta
\right) \right) \sin \left( \phi \right) ,r\cos \left( \theta \right)
\right) , \label{Eqn:3dtorus}
\end{equation}%
where $\theta $, $\phi $ are angles which make a full circle, and $r$ is the
radius of the tube known as the minor radius.
All observation
data are located on the torus with $r=1$. Then, the generalized
conditional density $p(r,\theta,\phi|D)$\ in $\left( r,\theta,\phi \right)$ coordinate can be defined using the Dirac delta function as follows,%
\begin{equation}
p\left(r, \theta ,\phi |D\right) =p_{\text{IC}}\left( \theta ,\phi
|D\right) \delta \left( r-1\right) , \label{Eqn:p123_r}
\end{equation}
where $p_{\text{IC}}$, defined in (\ref{Eqn:pic_trans}), denotes the conditional density function in the intrinsic coordinate, $\left( \theta ,\phi \right) $, and $%
\delta $ is the Dirac delta function.
After coordinate transformation, the density $p_{\text{EX}}:\mathbb{R}^3\to R$
can be obtained as
\begin{equation}
p_{\text{EX}}\left( \mathbf{x}|D\right) =\frac{p\left( r,\theta
,\phi |D\right) }{J}= \frac{p_{\text{IC}}\left( \theta ,\phi
|D\right) \delta \left( r-1\right) }{J}, \label{Eqn:phc}
\end{equation}
where $J$ is the Jacobian determinant $\det \left[ \partial
\left( x,y,z\right) /\partial \left( r,\theta ,\phi \right) \right] $.
It can be
examined that $p_{\text{EX}}\left( \mathbf{x}|D\right) $\ is a generalized
conditional density, that is, $\int_{\mathbb{R}^{3}}p_{\text{EX}}(
\mathbf{x}|D) d\mathbf{x}=1$.
Now it can be
clearly seen that due to the Dirac delta function $\delta
\left( r-1\right) $ in (\ref{Eqn:phc}), the density $p_{\text{EX}}(\mathbf{x}|D)$\ is no longer in
the weighted Hilbert space, $p_{\text{EX}}(\mathbf{x}|D)\notin \mathcal{H}_{q^{-1}}\left(
\mathcal{R}\right) $.
Consequently, the error estimation in
Theorem \ref{thm:error1}\ becomes invalid in controlling the error of the
conditional density.

Here, the key point
is that for Cosine
and Hermite representations, the volume integral is with respect to $d\mathbf{x}$. The complete basis functions are obtained from
tensor product of three sets of basis functions in $(x,y,z)$ coordinates. In
order to represent a conditional density function $p_{\text{EX}}(\mathbf{x}|D)$ defined only on an intrinsically 2D torus
domain, theoretically infinite number of basis functions are needed.
However, numerically only finite number of basis functions can be used. Then, the density $p_{\text{EX}}(\mathbf{x}|D)$ in (\ref{Eqn:phc})
cannot be well approximated for Hermite and Cosine representations (\ref%
{Eqn:pyb_repre}). Moreover, if only finite
number of Hermite or Cosine basis functions are used for representations,
typically Gibbs phenomenon can be observed, i.e., the Dirac delta function $\delta \left( r-1\right) $\ in (\ref{Eqn:phc})\
will be approximated by a function having a single tall spike at $%
r=1$ with some oscillations at two sides along the $r$ direction. On the other hand,
the data-driven basis functions obtained via the diffusion maps algorithm are smooth functions defined on the data manifold $\mathcal{N}$.
Therefore, while the Gibbs phenomenon still occurs in this spectral expansion, it is due to finite truncation in representing a positive smooth functions (densities) on the data manifold, and not due to the singularity
that occurs in the ambient direction as in (\ref{Eqn:phc}).


\end{document}